\documentclass[journal]{IEEEtran}

\usepackage{amsmath,amsfonts,amssymb,amsthm}
\usepackage{mathrsfs}
\usepackage{comment,blkarray}
\usepackage{multirow,bigdelim}
\usepackage{cite}
\usepackage{tikz}
\usetikzlibrary{arrows,automata}
\usepackage{verbatim}
\usepackage{graphicx}
\usepackage{cases}	
\usepackage{booktabs}
\usepackage{caption}												
\usepackage{etoolbox}

\usepackage{enumitem}
\usepackage{mathtools}
\usepackage{algorithm}
\usepackage{algorithmic}

\newtheorem{thm}{Theorem}
\newtheorem{lemma}{Lemma}
\newtheorem{cor}{Corollary}

\theoremstyle{definition}
\newtheorem{example}{Example}
\newtheorem{defn}{Definition}
\newtheorem{remark}{Remark}
\allowdisplaybreaks

\DeclareMathOperator{\lcm}{lcm}

\begin{document}

\title{Exact Dynamic Range--Robustness Tradeoff for Chinese Remainder Theorem under Non-Uniformly Bounded Remainder Errors}

\author{Guangpu Guo, \IEEEmembership{Graduate Student Member}, \IEEEmembership{IEEE}, and Xiang-Gen Xia, \IEEEmembership{Fellow}, \IEEEmembership{IEEE} 
       
\thanks{G. Guo and X.-G. Xia are with the Department of Electrical and Computer Engineering,
  University of Delaware, Newark, DE 19716, USA
  (e-mails: guangpu@udel.edu and xxia@ee.udel.edu).
  This work was supported in part by the National Science Foundation (NSF) under Grant CCF-2246917.
}
}

\maketitle

\begin{abstract}
The Chinese remainder theorem (CRT) is highly sensitive to remainder errors.
Robust CRT addresses this problem by recovering the folding integers correctly from erroneous remainders, so that the final reconstruction error is bounded by the remainder error level. 
The existing dynamic range--robustness tradeoff results mainly assume a uniform remainder error bound, where all remainders share one scalar error bound. 
This paper studies the exact tradeoff when the remainder error bound is a vector in the sense that different remainders may have different error bounds. 
For any fixed candidate range length, which specifies the integer interval over which an unknown integer is to be determined from its erroneous remainders, and any fixed remainder error bound vector, we derive a scalar necessary and sufficient condition for robust determination.
Based on this condition, we characterize the largest admissible range length, called the dynamic range, for a given remainder error bound
vector. 
Conversely, for a given range length, we characterize all remainder error bound vectors under which remainder errors can be tolerated. 
We also give an exact folding integer vector decoding algorithm. 
We show that the proposed results reduce to the classical CRT in the error-free case and coincide with the known result when all remainder error bounds are equal. 
Numerical results show that the dynamic range depends on the full remainder error bound vector. 
They also show that the remainders are coupled in the robustness, which cannot be captured by a single uniform error bound.
\end{abstract}


\section{Introduction}

The Chinese remainder theorem (CRT) is a classical result for reconstructing an integer from its remainders modulo several positive integer moduli. It has been widely used in number theory, coding theory, cryptography, and signal processing \cite{McClellanRader79,DingPeiSalomaa99,radar_book}. In the error-free case, a nonnegative integer can be uniquely determined over an integer interval whose length is the least common multiple (lcm) of all the moduli. This largest guaranteed interval length is called the \emph{dynamic range}. In this paper, a \emph{range length} means the length of a candidate integer interval over which an unknown integer is to be robustly determined from erroneous remainders.

A major problem of the classical CRT in applications, such as signal processing, is its sensitivity to remainder errors. A small error in one remainder may lead to a large reconstruction error. This issue appears naturally when remainders are obtained from noisy or unreliable observations, for example in phase unwrapping \cite{Xia99,XiaWang07,LiLiangXia09,YangWangXiaYin14,AkhlaqMcKilliamSubramanian15,LiWangWangMoran13}, distributed storage \cite{ChessaMaestrini12}, remainder coding \cite{BarsiMaestrini73,GoldreichRonSudan00,GuruswamiSahaiSudan00,ShparlinskiSteinfeld04,XiaoXia15Polynomial,XiaoHuangYeXiao18}, fault-tolerant analog neural network training \cite{DemirkiranNairBunandarJoshi24,DemirkiranYangBunandarJoshi24}, and multi-channel modulo sampling and modulo analog-to-digital converters (ADCs) \cite{GanLiu20,GongGanLiu21,YanLiGanLiuLi26,GuoBhandari26AsyncMCUSF}. Robust CRT addresses this problem by correctly determining the folding integers from erroneous remainders. Once the folding integers are correct, the final reconstruction error is bounded by the remainder error level \cite{XiaWang07,LiLiangXia09,WangXia10,XiaoXiaWang14}.

Early robust CRT results mainly keep the dynamic range equal to the classical CRT dynamic range and study the largest tolerable remainder error level. Closed-form reconstruction algorithms and exact conditions were obtained in the common-factor co-prime setting \cite{WangXia10}, and later extended to general moduli and multi-stage robust CRT \cite{XiaoXiaWang14}. Beyond the single integer case, CRT and robust CRT have also been extended to real-valued and complex-valued unknowns \cite{WangLiWangXia15,xiaopingli}, multiple integers \cite{Xia99,XiaoHuangYeXiao18,liao07,li16,XiaoXiao19}, and multidimensional integer vectors \cite{XiaoXiaWang20,XiaoHuoXia24,GuoXia25a,GuoXia25b,mstage-mdcrt}. 

A different and fundamental question arises when the dynamic range is allowed to change. A larger dynamic range allows more integers to be reconstructed, but it also makes the remainder patterns denser and leaves less room for remainder errors. Thus, robustness can be improved by introducing redundancy, or equivalently by creating more separation among remainder patterns. This can be done either by using redundant moduli, such as moduli with a common factor \cite{XiaWang07,WangXia10,XiaoXiaWang14}, or by keeping the moduli fixed and reducing the dynamic range considered for reconstruction \cite{Parhami15,XiaoXiaHuo17,XiaoHuangYeXiao18,YanGanLiuHu25}. In particular, pairwise co-prime moduli may still provide robustness when the considered dynamic range is reduced, even if every remainder is allowed to contain a bounded error.

The idea of reducing dynamic range is related to, but different from, remainder error correction codes based on redundant residue number systems (RRNS). In a typical RRNS code with $n$ pairwise co-prime moduli and $k<n$, the message range is chosen as the product of the $k$ smallest moduli, and the remaining remainders provide redundancy so that the message can be recovered from any known set of at least $k$ error-free remainders \cite{BarsiMaestrini73,GoldreichRonSudan00,GuruswamiSahaiSudan00,ShparlinskiSteinfeld04}, where the errors in the remaining remainders can be arbitrary. In contrast, this paper studies robust determination in the robust CRT sense, where every remainder may be erroneous, but each error is bounded. 

For two-modulus systems, the dynamic range--robustness tradeoff was first investigated through a position representation on the two-dimensional remainder plane \cite{Parhami15}, and exact characterization and closed-form reconstruction algorithms were later derived \cite{XiaoXiaHuo17}. For general moduli, a geometric framework based on a pseudo metric was proposed to exactly characterize the tradeoff under a uniform remainder error bound \cite{XiaoHuangYeXiao18}. More recently, this tradeoff has also been studied from the viewpoint of moduli design \cite{YanGanLiuHu25}.

The existing dynamic range--robustness tradeoff results are mainly formulated under a uniform remainder error model, where all remainders share the same error bound. This assumption may be restrictive in applications, since different remainder channels may have different noise levels or reliability requirements, for example in distributed storage, multi-channel sampling, multi-wavelength sensing, and modulo ADCs. A single uniform bound may either underuse reliable channels or overstate the tolerance of unreliable ones. Therefore, a natural question is to characterize the tradeoff when the remainder error bound is a vector rather than a scalar, that is, when different remainder errors may have different bounds (levels). This leads to the two central questions of this paper. For a given remainder error bound vector, what is the dynamic range for robust recovery? Conversely, for a given range length, which remainder error bound vectors can be tolerated? This paper gives exact answers to both questions.
 
We adopt a discrete remainder error model, where the observed remainders and remainder errors are integer-valued. This model is natural in residue number systems and Chinese remainder codes, and it also appears when continuous measurements are quantized or detected into discrete bins, such as frequency bins or range bins. Hence the remainder error bounds are represented by an integer vector.

To answer the above questions, we view robust CRT as an exact folding integer vector decoding problem. Here, a folding integer vector is the integer vector formed by the folding integers associated with the different moduli. For a fixed range length and a fixed remainder error bound vector, we collect all possible observed remainder vectors that can be produced by integers with the same folding integer vector. Robust determination is possible if and only if the collections corresponding to different folding integer vectors are disjoint. From this observation, we derive a scalar necessary and sufficient condition for robust determination at any fixed range length and remainder error bound vector. This condition is the main tool for characterizing both the dynamic range for a given remainder error bound vector and the tolerable remainder error bound vectors for a given range length.

For a given remainder error bound vector, we show that the dynamic range is determined by the earliest ambiguity that appears as the range length increases. Here, an ambiguity means that two integers with different folding integer vectors can generate the same erroneous remainder vector under the allowed remainder error levels. For a fixed range length, we characterize all tolerable remainder error bound vectors through a finite set of boundary vectors. These boundary vectors are the minimal remainder error bound vectors that can cause an ambiguity. Hence, a remainder error bound vector is tolerable if and only if it is not componentwise larger than or equal to any boundary vector.

We also give an exact decoding algorithm for the folding integer vector. Given an erroneous remainder vector, the proposed algorithm searches for a candidate integer in the tested range that could have generated the observation under the given remainder error bounds. When the proposed necessary and sufficient condition for robust determination holds, all the candidate integers that can generate the same observation have the same folding integer vector, and hence the algorithm returns the desired folding integer vector.

We further discuss special cases to connect the proposed tradeoff with the existing CRT and robust CRT results. In the error-free case, the dynamic range reduces to the classical CRT dynamic range. In the uniform remainder error bound case, our result coincides with the known uniform tradeoff in \cite{XiaoHuangYeXiao18}. Numerical results verify the decoding algorithm, illustrate that pairwise co-prime moduli can still provide robustness over a reduced dynamic range, and show that the full remainder error bound vector is needed. In particular, vectors with the same largest component may have different dynamic ranges, and the tolerable remainder error bounds in different remainder channels are coupled.

The rest of this paper is organized as follows. Section~\ref{s2} gives the problem setup. Section~\ref{s3} develops the exact dynamic range--robustness tradeoff. Section~\ref{s4} presents the exact folding integer vector decoding algorithm. Section~\ref{s5} discusses special cases and connections with the known robust CRT results. Section~\ref{s6} presents numerical results. Section~\ref{s7} concludes the paper.

\emph{Notation:}
Let $\mathbb Z$, $\mathbb Z_{\ge0}$, and $\mathbb Z_{>0}$ denote the sets of integers, nonnegative integers, and positive integers, respectively.
For integers $a$ and $b$, let $[a,b]_{\mathbb Z}=\{a,a+1,\ldots,b\}$, with the convention that $[a,b]_{\mathbb Z}=\emptyset$ if $a>b$. 
For a real number $x$, $[x]$ denotes a rounding integer satisfying $-\frac{1}{2} \le x-[x] < \frac{1}{2}$, while $\lfloor x\rfloor$ and $\lceil x\rceil$ denote the floor and ceiling of $x$, respectively. 
Vectors are regarded as column vectors. 
We use $\mathbf 0$ and $\mathbf 1$ for the all-zero and all-one vectors of proper dimensions. 
For $\boldsymbol{\tau},\boldsymbol{\eta}\in\mathbb Z_{\ge0}^L$, we write $\boldsymbol{\tau}\preceq\boldsymbol{\eta}$ if $\tau_i\le\eta_i$ for every $i=1,\ldots,L$. 
Equivalently, $\boldsymbol{\eta}$ is said to dominate $\boldsymbol{\tau}$.

\section{Problem Setup}\label{s2}

Let $m_1,\ldots,m_L$ be arbitrary positive integers, called integer moduli, and denote their greatest common divisor (gcd) and least common multiple (lcm) by
\begin{equation}\label{eq:gcd_lcm}
m=\gcd(m_1,\ldots,m_L),
\qquad
M=\lcm(m_1,\ldots,m_L).
\end{equation}
For any integer $N\in\mathbb Z$, there exist unique integers
$n_i\in\mathbb Z$ and $r_i\in[0,m_i-1]_{\mathbb Z}$ such that
\begin{equation}\label{eq:cong}
N=n_i m_i+r_i,
\qquad
i=1,\ldots,L.
\end{equation}
Here, $n_i$ is called the \emph{folding integer}, and $r_i$ is called the
\emph{remainder} of $N$ modulo $m_i$.
By the classical CRT, an integer $N$ can be uniquely determined from its error-free remainders $r_1,\ldots,r_L$ modulo $m_1,\ldots,m_L$, respectively, if and only if $N\in [0,M-1]_{\mathbb Z}$.
This interval has length $M$, and no longer interval can guarantee unique determination. In this case, $M$ is called the \emph{dynamic range}.

We next introduce the remainder error model. In practice, the remainders may not be obtained exactly. Instead, for each modulus $m_i$, $i=1,\ldots,L$, we observe an erroneous remainder
\begin{equation}\label{eq:error_model}
\widetilde r_i=r_i+\Delta r_i \in[0,m_i-1]_{\mathbb Z},
\qquad
\Delta r_i\in\mathbb Z,
\end{equation}
where $\Delta r_i$ denotes the error in the $i$-th remainder.

In this paper, we use a discrete remainder error model. That is, the observed remainders $\widetilde r_i$ are also integer-valued, so the error $\Delta r_i$ is an integer. This model is natural when remainders are represented as digital symbols or detected indices. For example, in residue number systems and Chinese remainder codes, each remainder coordinate is an integer. In many signal processing applications, a continuous measurement is first quantized or detected into a bin, such as a frequency bin or range bin. The resulting remainder estimate is integer-valued, and the corresponding remainder error is an integer. Furthermore, in the case of fractional numbers, one can always multiply them by a common factor to make the results all integers.

Different remainders may have different error bounds. Let
\begin{equation}\label{eq:tau_vector}
\boldsymbol{\tau}=(\tau_1,\ldots,\tau_L)\in\mathbb Z_{\ge 0}^L
\end{equation}
be a remainder error bound vector:
\begin{equation}\label{eq:nonuniform_error_bound}
|\Delta r_i|\le \tau_i,
\qquad
i=1,\ldots,L.
\end{equation}

For an unknown integer $N$, we say that $N$ can be \emph{robustly determined} from the erroneous remainders $\widetilde r_1,\ldots,\widetilde r_L$ if the folding integers $n_1,\ldots,n_L$ in \eqref{eq:cong} can be correctly determined.
This definition follows the standard robust CRT viewpoint \cite{XiaWang07,LiLiangXia09,WangXia10,XiaoXiaWang14}. Once the folding integers are correctly determined, one can estimate $N$ by
\begin{equation}\label{eq:N_hat_average}
\widehat N
=
\left[
\frac{1}{L}\sum_{i=1}^{L}(\widehat n_i m_i+\widetilde r_i)
\right],
\end{equation}
where $\widehat n_i=n_i$ for $i=1,\ldots,L$. Then, naturally,
\begin{equation}\label{eq:error_N}
|\widehat N-N|
\le
\left[ \frac{1}{L}\sum_{i=1}^{L}\tau_i \right].
\end{equation}
Thus, once the folding integers are correctly determined, bounded remainder errors lead to a bounded reconstruction error. Therefore, the essential task in robust recovery here is to determine the folding integers correctly.

\begin{remark}[Folding integer recovery versus exact integer recovery]
The robust determination considered in this paper should not be confused with exact decoding in Chinese remainder codes. Here we do not require the erroneous remainders to identify the exact unknown integer $N$ without error. We only require them to identify the folding integers correctly, which is the standard robust CRT viewpoint.
\end{remark}

We now formally define the dynamic range in terms of a remainder error bound vector.

\begin{defn}[Dynamic range]\label{defn:dynamic_range}
For a given remainder error bound vector $\boldsymbol{\tau}$ defined in \eqref{eq:tau_vector}, the \emph{dynamic range} under $\boldsymbol{\tau}$, denoted by $D(\boldsymbol{\tau})$, is the largest positive integer $D$ such that every integer $N\in[0,D-1]_{\mathbb Z}$ can be robustly determined, i.e., its folding integers $n_1,\ldots,n_L$ in \eqref{eq:cong} can be correctly determined, from the erroneous remainders $\widetilde r_1,\ldots,\widetilde r_L$ for the remainder errors satisfying \eqref{eq:nonuniform_error_bound}.
\end{defn}

Consider the componentwise partial order $\preceq$ on $\mathbb Z_{\ge 0}^L$. If $\boldsymbol{\tau}\preceq\boldsymbol{\eta}$, then, clearly we have
$D(\boldsymbol{\eta})\le D(\boldsymbol{\tau})$.
Also, for any remainder error bound vector
$\boldsymbol{\tau}\in\mathbb Z_{\ge 0}^L$, we have
$D(\boldsymbol{\tau})\le M$, because $\mathbf 0\preceq\boldsymbol{\tau}$ and the zero-error case has dynamic range $M$.

In the following, a \emph{range length} means an integer $D$ with
$1\le D\le M$ for which we test the robust determination on
$[0,D-1]_{\mathbb Z}$ using the erroneous remainders $\widetilde r_i$. Such a $D$ is a candidate length, not necessarily the
dynamic range. It is said admissible under $\boldsymbol{\tau}$ if the robust
determination is guaranteed for integers in $[0,D-1]_{\mathbb Z}$ from erroneous remainders $\widetilde r_1,\ldots,\widetilde r_L$ in \eqref{eq:tau_vector}--\eqref{eq:nonuniform_error_bound}. Hence
$D(\boldsymbol{\tau})$ is the largest admissible range length.

The above definition leads to two closely related problems studied in this paper. First, for a given remainder error bound vector $\boldsymbol{\tau}$, we aim to characterize the dynamic range $D(\boldsymbol{\tau})$. In other words, we want to determine the largest integer interval $[0,D(\boldsymbol{\tau})-1]_{\mathbb Z}$ over which robust determination is guaranteed for all remainder errors satisfying \eqref{eq:nonuniform_error_bound}.
Second, for a given range length $D$, we aim to characterize the tolerable remainder error bound vectors $\boldsymbol{\tau}$ such that $D\le D(\boldsymbol{\tau})$. Equivalently, we want to determine how large the componentwise error bounds $\tau_1,\ldots,\tau_L$ can be while still guaranteeing the robust determination of every integer $N\in[0,D-1]_{\mathbb Z}$.
Together, these two problems describe the exact tradeoff between dynamic
range and robustness under bounded remainder errors.

\begin{remark}[Connection with RRNS-based remainder error correction codes]
The tradeoff considered in this paper is related to remainder error correction codes based on RRNS \cite{BarsiMaestrini73,GoldreichRonSudan00,GuruswamiSahaiSudan00,ChessaMaestrini12,ShparlinskiSteinfeld04}.
In a typical RRNS code, one uses $n$ moduli, but chooses the message range so that the message can be recovered from any $k$ error-free remainders for $k<n$.
If recovery is required from any known set $S$ of $k$ error-free remainders, then the message range cannot be larger than $\min_{|S|=k}\lcm\{m_i:i\in S\}$.
For pairwise co-prime moduli, this leads to the choice where the message range is determined by the smallest $k$ moduli, and the remaining remainders provide redundancy for error correction, i.e., error-free recovery.

The setting of this paper is different. We do not assume that a fixed number
of remainders are error-free. Instead, every remainder may contain an error, and
the allowable errors are described by the remainder error bound vector
$\boldsymbol{\tau}$. Therefore, the largest admissible range length $D(\boldsymbol{\tau})$ is not
determined by a fixed number $k$ of reliable remainders but depends on the full
remainder error bound vector $\boldsymbol{\tau}$. 
Also, this paper only considers the robust recovery but not error-free recovery as in error correction coding, although both trade the dynamic range with the error tolerance (robustness and error correction capability).
\end{remark}

\section{Exact Dynamic Range--Robustness Tradeoff}\label{s3}

In this section, we characterize the exact tradeoff between the dynamic range and the remainder error bounds. The main idea is to group the unknown integers according to their folding integer vectors.
For each group, we describe the
set of all possible erroneous remainder vectors that can be generated by
integers in this group, and call it the observation set of the corresponding
folding integer vector.
The robust recovery is possible only when the observation sets corresponding to different folding integer vectors are disjoint.
Based on this viewpoint, we derive an exact separation condition. We then use this condition to characterize both the dynamic range
$D(\boldsymbol{\tau})$ for a given remainder error bound vector $\boldsymbol{\tau}$ and the error bounds that can be tolerated for a given range length $D$.

\subsection{Observation Sets Induced by Folding Integer Vectors}
\label{s3_sub1}

For an integer $N$, its folding integers $n_i$ in \eqref{eq:cong} are
\begin{equation}\label{eq:n_r_components}
n_i(N)=\left\lfloor \frac{N}{m_i}\right\rfloor,
\qquad i=1,\ldots,L.
\end{equation}
Define the \emph{folding integer vector} of $N$ by
\begin{equation}\label{eq:folding_vector}
\mathbf n(N)
=
(n_1(N),\ldots,n_L(N)).
\end{equation}
For a range length $D$ with $1\le D\le M$, define the set of all folding integer vectors of the integers in $[0,D-1]_{\mathbb Z}$ by
\begin{equation}\label{eq:folding_vector_set}
\mathcal F(D)
=
\{\mathbf n(N):N\in[0,D-1]_{\mathbb Z}\}.
\end{equation}

It is important to note that the map $N\mapsto \mathbf n(N)$ is not necessarily one-to-one on $[0,D-1]_{\mathbb Z}$. Indeed, by \eqref{eq:n_r_components}, the component $n_i(N)$ remains unchanged as long as $N$ stays in an integer interval of the form
\begin{equation}
[qm_i,(q+1)m_i-1]_{\mathbb Z},
\qquad q\in\mathbb Z_{\ge0}.
\end{equation}
Since these intervals depend on the modulus $m_i$, several consecutive integers may share the same folding integer vector. Therefore, it is natural to group together all the integers that have the same folding integer vector.

For each $\mathbf n\in\mathcal F(D)$, define the set of integers corresponding to the folding integer vector $\mathbf n$ by
\begin{equation}\label{eq:folding_set_def}
\mathcal I_D(\mathbf n)
=
\{N\in[0,D-1]_{\mathbb Z}:\mathbf n(N)=\mathbf n\}.
\end{equation}
Thus, $\mathcal I_D(\mathbf n)$ contains exactly those integers in $[0,D-1]_{\mathbb Z}$ that have the same folding integer vector $\mathbf n$.

As discussed above, for a fixed folding integer vector $\mathbf n=(n_1,\ldots,n_L)\in\mathcal F(D)$, the condition $\mathbf n(N)=\mathbf n$ is equivalent to
\[
n_i m_i\le N\le (n_i+1)m_i-1,
\qquad i=1,\ldots,L.
\]
Therefore,
\begin{equation}\label{eq:folding_set_intersection}
\mathcal I_D(\mathbf n)
=
[0,D-1]_{\mathbb Z}
\cap
\bigcap_{i=1}^{L}
[n_i m_i,(n_i+1)m_i-1]_{\mathbb Z}.
\end{equation}
Since $\mathbf n\in\mathcal F(D)$, this set is nonempty. Hence it can be written as
\begin{equation}\label{eq:folding_set_interval}
\mathcal I_D(\mathbf n)
=
[\alpha_D(\mathbf n),\beta_D(\mathbf n)]_{\mathbb Z},
\end{equation}
where
\begin{align}
\alpha_D(\mathbf n)
&=
\max_{1\le i\le L} n_i m_i,
\label{eq:alpha_def}\\
\beta_D(\mathbf n)
&=
\min\left\{
D-1,\min_{1\le i\le L}\bigl((n_i+1)m_i-1\bigr)
\right\}.
\label{eq:beta_def}
\end{align}

\begin{example}\label{ex:folding_vector_example}
Let $(m_1,m_2)=(6,10)$ and $D=12$. Then $\mathcal F(12)=\{(0,0),(1,0),(1,1)\}$. The corresponding integer sets are $\mathcal I_{12}(0,0)=[0,5]_{\mathbb Z}$, $\mathcal I_{12}(1,0)=[6,9]_{\mathbb Z}$, and $\mathcal I_{12}(1,1)=[10,11]_{\mathbb Z}$. Fig.~\ref{fig:folding_vector_example} plots the remainder vectors $(r_1(N),r_2(N))$ for $N\in[0,11]_{\mathbb Z}$. The points with the same color share the same folding integer vector. This example shows that several consecutive integers may have the same folding integer vector.

\begin{figure}[htbp]
\centering
\includegraphics[width=0.18\textwidth]{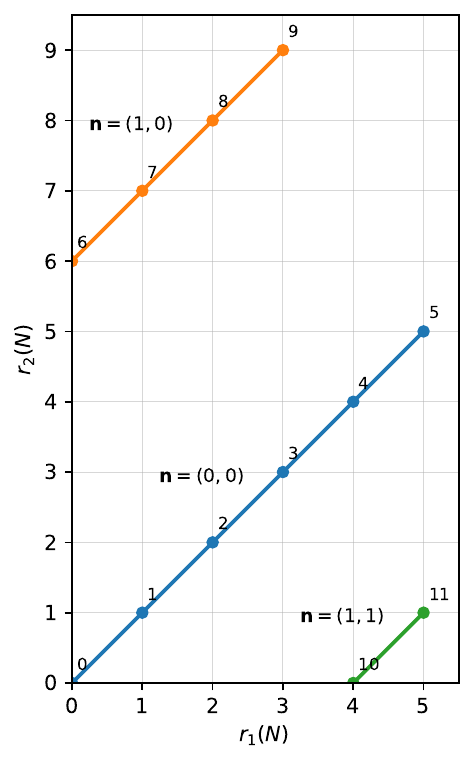}
\caption{Remainder vectors $(r_1(N),r_2(N))$ for $(m_1,m_2)=(6,10)$ and $D=12$.}
\label{fig:folding_vector_example}
\end{figure}

\end{example}
We next describe the possible erroneous remainder vectors generated by a single integer. To simplify the notation, define the set of admissible integer error vectors under the error bound vector $\boldsymbol{\tau}$ by
\begin{equation}\label{eq:admissible_error_set}
\mathcal E_{\boldsymbol{\tau}}
=
\left\{
\mathbf e=(e_1,\ldots,e_L)\in\mathbb Z^L:
|e_i|\le \tau_i,\ i=1,\ldots,L
\right\}.
\end{equation}

For any integer $N$, its $i$-th true remainder is
\[
r_i(N)=N-n_i(N)m_i,
\qquad i=1,\ldots,L.
\]
Equivalently, the true \emph{remainder vector} of $N$ is
\begin{equation}\label{eq:true_remainder_vector}
\mathbf r(N)
=
(r_1(N),\ldots,r_L(N))
=
N\mathbf 1-\mathbf m\odot\mathbf n(N),
\end{equation}
where $\mathbf m=(m_1,\ldots,m_L)$, and
$\mathbf m\odot\mathbf n(N)=(m_1n_1(N),\ldots,m_Ln_L(N))$.
Under the remainder error model \eqref{eq:error_model}--\eqref{eq:nonuniform_error_bound}, the observed remainder vector generated by $N$ has the form
\begin{equation}\label{eq:observed_remainder_vector_single_N}
\widetilde{\mathbf r}
=
\mathbf r(N)+\mathbf e
=
N\mathbf 1-\mathbf m\odot\mathbf n(N)+\mathbf e,
\qquad
\mathbf e\in\mathcal E_{\boldsymbol{\tau}}.
\end{equation}

Now fix a range length $D$ with $1\le D\le M$ and an error bound vector $\boldsymbol{\tau}$. For each folding integer vector $\mathbf n\in\mathcal F(D)$, we define the observation set generated by $\mathbf n$ as
\begin{equation}\label{eq:observation_set_def}
\begin{aligned}
\mathcal O_{\boldsymbol{\tau},D}(\mathbf n)
=
\{\,\widetilde{\mathbf r}\in\mathbb Z^L:
&\ \widetilde{\mathbf r}=\mathbf r(N)+\mathbf e,\\
&\ N\in\mathcal I_D(\mathbf n),\
\mathbf e\in\mathcal E_{\boldsymbol{\tau}}\,\}.
\end{aligned}
\end{equation}
By \eqref{eq:error_model}, an actually observable remainder vector must also belong to $\mathcal A\triangleq\prod_{i=1}^{L}[0,m_i-1]_{\mathbb Z}$.
Hence the physically admissible observation set is $\mathcal O_{\boldsymbol{\tau},D}(\mathbf n)\cap\mathcal A$.
As will be shown in the following lemma, restricting the observation vectors to $\mathcal A$ does not affect the pairwise disjointness condition. 
Therefore, for simplicity, we continue to use $\mathcal O_{\boldsymbol{\tau},D}(\mathbf n)$ in \eqref{eq:observation_set_def} as the observation set in the subsequent analysis.

We now connect the observation sets with robust recovery. By the definition in Section~\ref{s2}, robust determination means that the folding integers, or equivalently the folding integer vector, can be correctly determined from the erroneous remainders. Therefore, for fixed $D$ and $\boldsymbol{\tau}$, every possible observation $\widetilde{\mathbf r}$ must identify a unique folding integer vector.
Equivalently, the observation sets corresponding to different folding integer vectors must be disjoint.

\begin{lemma}\label{lem:observation_set_criterion}
For a given $1\le D\le M$, every integer $N\in[0,D-1]_{\mathbb Z}$ can be robustly determined under the remainder error bound vector $\boldsymbol{\tau}$ if and only if
\begin{equation}\label{eq:observation_set_disjoint}
\mathcal O_{\boldsymbol{\tau},D}(\mathbf n)
\cap
\mathcal O_{\boldsymbol{\tau},D}(\mathbf n')
=
\emptyset
\end{equation}
for every two distinct folding integer vectors $\mathbf n,\mathbf n'\in\mathcal F(D)$.
\end{lemma}

\begin{proof}
We first note that restricting the observation sets to $\mathcal A=\prod_{i=1}^{L}[0,m_i-1]_{\mathbb Z}$ does not change their pairwise disjointness.
If two observation sets $\mathcal O_{\boldsymbol{\tau},D}(\mathbf n)$ and $\mathcal O_{\boldsymbol{\tau},D}(\mathbf n')$ overlap, then for some $N$ and $N'$ with the corresponding folding integer vectors $\mathbf n$ and $\mathbf n'$, $|r_i(N)-r_i(N')|\le 2\tau_i$ for every $i$.
Choosing an integer $\bar r_i=[(r_i(N)+r_i(N'))/2]$ gives $|\bar r_i-r_i(N)|\le\tau_i$ and $|\bar r_i-r_i(N')|\le\tau_i$. 
Moreover, $\bar r_i$ lies between $r_i(N)$ and $r_i(N')$, and hence $\bar r_i\in[0,m_i-1]_{\mathbb Z}$. 
Therefore, the two observation sets also overlap inside $\mathcal A$.
The converse is immediate.

Thus, \eqref{eq:observation_set_disjoint} is equivalent to requiring
the physically admissible observation sets $\mathcal O_{\boldsymbol{\tau},D}(\mathbf n)\cap\mathcal A$ and $\mathcal O_{\boldsymbol{\tau},D}(\mathbf n')\cap\mathcal A$ to be pairwise disjoint.
This holds if and only if every admissible erroneous remainder vector identifies a unique folding integer vector, which is exactly the robust determination.
\end{proof}

\subsection{Overlap of Observation Sets}
\label{s3_sub2}

We next derive a simple condition for the overlap of two observation sets. Fix two distinct folding integer vectors
$\mathbf n,\mathbf n'\in\mathcal F(D)$ with $\mathbf n\ne\mathbf n'$.
By the definition of $\mathcal O_{\boldsymbol{\tau},D}(\mathbf n)$ in \eqref{eq:observation_set_def}, the two observation sets
$\mathcal O_{\boldsymbol{\tau},D}(\mathbf n)$ and
$\mathcal O_{\boldsymbol{\tau},D}(\mathbf n')$ overlap if and only if there exist
$N\in\mathcal I_D(\mathbf n)$, $N'\in\mathcal I_D(\mathbf n')$, and
$\mathbf e,\mathbf e'\in\mathcal E_{\boldsymbol{\tau}}$ such that
\begin{equation}\label{eq:two_observation_equal_vector}
\mathbf r(N)+\mathbf e
=
\mathbf r(N')+\mathbf e'.
\end{equation}
Thus, the problem is to decide whether the difference between the two true
remainder vectors can be written as $\mathbf e-\mathbf e'$ for some
$\mathbf e,\mathbf e'\in\mathcal E_{\boldsymbol{\tau}}$.

Since $N\in\mathcal I_D(\mathbf n)$ and $N'\in\mathcal I_D(\mathbf n')$, we have $\mathbf n(N)=\mathbf n$ and $\mathbf n(N')=\mathbf n'$. Hence, for each $i=1,\ldots,L$,
\begin{equation}\label{eq:remainder_difference_first}
\begin{aligned}
r_i(N)-r_i(N')
&=
\bigl(N-n_i m_i\bigr)
-
\bigl(N'-n_i'm_i\bigr)\\
&=
(N-N')-(n_i-n_i')m_i.
\end{aligned}
\end{equation}
This decomposition separates the difference into two parts. The term $N-N'$ depends on the particular integers chosen from the two integer sets, while the term $(n_i-n_i')m_i$ depends only on the two folding integer vectors. Define
\begin{equation}\label{eq:s_def}
s=N-N',
\end{equation}
and
\begin{equation}\label{eq:c_i_def}
c_i(\mathbf n,\mathbf n')
=
(n_i-n_i')m_i,
\qquad
i=1,\ldots,L.
\end{equation}
Then \eqref{eq:remainder_difference_first} becomes
\begin{equation}\label{eq:remainder_difference_s}
r_i(N)-r_i(N')
=
s-c_i(\mathbf n,\mathbf n'),
\qquad
i=1,\ldots,L.
\end{equation}

Let
$\mathcal J_D(\mathbf n,\mathbf n') =
\mathcal I_D(\mathbf n)-\mathcal I_D(\mathbf n')$
be the difference set of the two intervals
$\mathcal I_D(\mathbf n)$ and $\mathcal I_D(\mathbf n')$. By
\eqref{eq:folding_set_interval},
$\mathcal I_D(\mathbf n)
=
[\alpha_D(\mathbf n),\beta_D(\mathbf n)]_{\mathbb Z}$ and 
$\mathcal I_D(\mathbf n')
=
[\alpha_D(\mathbf n'),\beta_D(\mathbf n')]_{\mathbb Z}$.
Hence
\begin{equation}\label{eq:J_D_def}
\mathcal J_D(\mathbf n,\mathbf n')=
[\alpha_D(\mathbf n)-\beta_D(\mathbf n'),
\beta_D(\mathbf n)-\alpha_D(\mathbf n')]_{\mathbb Z},
\end{equation}
which is exactly the set of all possible
values of $s=N-N'$ with $N\in\mathcal I_D(\mathbf n)$ and
$N'\in\mathcal I_D(\mathbf n')$.

We next use the remainder error bound vector $\boldsymbol{\tau}$ to further restrict the
possible values of $s$. If the two observation sets overlap, then
\eqref{eq:two_observation_equal_vector} holds for some
$\mathbf e,\mathbf e'\in\mathcal E_{\boldsymbol{\tau}}$. Hence, for each
$i=1,\ldots,L$,
$r_i(N)-r_i(N')
=
e_i'-e_i$.
Together with \eqref{eq:remainder_difference_s}, this gives
\[
s-c_i(\mathbf n,\mathbf n')
=
e_i'-e_i.
\]
Since $|e_i|\le \tau_i$ and $|e_i'|\le \tau_i$, we must have
\begin{equation}\label{eq:coordinate_error_compatibility}
|s-c_i(\mathbf n,\mathbf n')|
\le
2\tau_i,
\qquad i=1,\ldots,L.
\end{equation}
Thus, for each modulo index $i$, the value of $s$ must lie in the integer
interval
\[
[c_i(\mathbf n,\mathbf n')-2\tau_i,\,
c_i(\mathbf n,\mathbf n')+2\tau_i]_{\mathbb Z}.
\]
Since the condition must hold for all modulo indices, define
\begin{equation}\label{eq:W_tau_def}
\mathcal W_{\boldsymbol{\tau}}(\mathbf n,\mathbf n')
=
\bigcap_{i=1}^{L}
[c_i(\mathbf n,\mathbf n')-2\tau_i,\,
c_i(\mathbf n,\mathbf n')+2\tau_i]_{\mathbb Z}.
\end{equation}
Equivalently,
\begin{equation}\label{eq:W_pq_def}
\mathcal W_{\boldsymbol{\tau}}(\mathbf n,\mathbf n')
=
[p_{\boldsymbol{\tau}}(\mathbf n,\mathbf n'),
q_{\boldsymbol{\tau}}(\mathbf n,\mathbf n')]_{\mathbb Z},
\end{equation}
where
\begin{equation}\label{eq:pq_tau_def}
\begin{aligned}
p_{\boldsymbol{\tau}}(\mathbf n,\mathbf n')
&=
\max_{1\le i\le L}
\left(
c_i(\mathbf n,\mathbf n')-2\tau_i
\right),\\
q_{\boldsymbol{\tau}}(\mathbf n,\mathbf n')
&=
\min_{1\le i\le L}
\left(
c_i(\mathbf n,\mathbf n')+2\tau_i
\right).
\end{aligned}
\end{equation}
Therefore, if the two observation sets overlap, then the corresponding value
$s=N-N'$ must belong to
$\mathcal W_{\boldsymbol{\tau}}(\mathbf n,\mathbf n')$.

The argument above shows that if the two observation sets overlap, then the corresponding difference $s=N-N'$ must belong to both $\mathcal J_D(\mathbf n,\mathbf n')$ and $\mathcal W_{\boldsymbol{\tau}}(\mathbf n,\mathbf n')$. The next lemma shows that this necessary condition is also sufficient.

\begin{lemma}\label{lem:interval_overlap}
For two distinct folding integer vectors $\mathbf n,\mathbf n'\in\mathcal F(D)$,
\begin{equation}\label{eq:observation_intersection_nonempty}
\mathcal O_{\boldsymbol{\tau},D}(\mathbf n)
\cap
\mathcal O_{\boldsymbol{\tau},D}(\mathbf n')
\ne
\emptyset
\end{equation}
if and only if
\begin{equation}\label{eq:J_W_intersection_condition}
\mathcal J_D(\mathbf n,\mathbf n')
\cap
\mathcal W_{\boldsymbol{\tau}}(\mathbf n,\mathbf n')
\ne
\emptyset.
\end{equation}
\end{lemma}

\begin{proof}
The necessity follows from the argument above. We prove the sufficiency.

Assume that
$s\in
\mathcal J_D(\mathbf n,\mathbf n')
\cap
\mathcal W_{\boldsymbol{\tau}}(\mathbf n,\mathbf n')$.
Since $s\in\mathcal J_D(\mathbf n,\mathbf n')$, there exist
$N\in\mathcal I_D(\mathbf n)$ and $N'\in\mathcal I_D(\mathbf n')$ such that
$s=N-N'$.

For each $i=1,\ldots,L$, define
$\eta_i
=
s-c_i(\mathbf n,\mathbf n')$.
Since $s\in\mathcal W_{\boldsymbol{\tau}}(\mathbf n,\mathbf n')$, we have
$|\eta_i|\le 2\tau_i$, for all $i=1,\ldots,L$.
For each $i$, choose
\[
e_i'
=
\left\lceil\frac{\eta_i}{2}\right\rceil,
\qquad
e_i
=
-\left\lfloor\frac{\eta_i}{2}\right\rfloor.
\]
Then
$e_i'-e_i
= \eta_i$.
Since $\eta_i$ and $\tau_i$ are integers and $|\eta_i|\le 2\tau_i$, this
choice also satisfies
$|e_i|\le \tau_i$ and 
$|e_i'|\le \tau_i$.

Using \eqref{eq:remainder_difference_s}, we obtain
\[
r_i(N)-r_i(N')
=
s-c_i(\mathbf n,\mathbf n')
=
\eta_i
=
e_i'-e_i.
\]
Hence
$r_i(N)+e_i
=
r_i(N')+e_i'$, for $i=1,\ldots,L$.
Let $\mathbf e=(e_1,\ldots,e_L)$ and
$\mathbf e'=(e_1',\ldots,e_L')$. Then
$\mathbf e,\mathbf e'\in\mathcal E_{\boldsymbol{\tau}}$ and
\[
\mathbf r(N)+\mathbf e
=
\mathbf r(N')+\mathbf e'.
\]
By the definition of the observation sets in
\eqref{eq:observation_set_def}, this implies \eqref{eq:observation_intersection_nonempty}.
This proves the sufficiency.
\end{proof}

\subsection{Separation Margin for the Exact Tradeoff}
\label{s3_sub3}

The interval test above gives an exact criterion for the overlap of two
observation sets. We now rewrite this criterion as a scalar separation margin.

Fix two distinct folding integer vectors $\mathbf n,\mathbf n'\in\mathcal F(D)$. By \eqref{eq:J_D_def} and \eqref{eq:W_pq_def}, both $\mathcal J_D(\mathbf n,\mathbf n')$ and $\mathcal W_{\boldsymbol{\tau}}(\mathbf n,\mathbf n')$ are integer intervals. Hence their intersection is also an integer interval. Define
\begin{equation}\label{eq:ell_u_margin_def}
\begin{aligned}
\ell_{\boldsymbol{\tau},D}(\mathbf n,\mathbf n')
&=
\max\left\{
\alpha_D(\mathbf n)-\beta_D(\mathbf n'),
p_{\boldsymbol{\tau}}(\mathbf n,\mathbf n')
\right\},\\
u_{\boldsymbol{\tau},D}(\mathbf n,\mathbf n')
&=
\min\left\{
\beta_D(\mathbf n)-\alpha_D(\mathbf n'),
q_{\boldsymbol{\tau}}(\mathbf n,\mathbf n')
\right\}.
\end{aligned}
\end{equation}
Then
\begin{equation}\label{eq:JW_as_interval}
\mathcal J_D(\mathbf n,\mathbf n')
\cap
\mathcal W_{\boldsymbol{\tau}}(\mathbf n,\mathbf n')
=
[\ell_{\boldsymbol{\tau},D}(\mathbf n,\mathbf n'),
u_{\boldsymbol{\tau},D}(\mathbf n,\mathbf n')]_{\mathbb Z}.
\end{equation}
This motivates the separation margin
\begin{equation}\label{eq:Theta_pairwise_def}
\Theta_{\boldsymbol{\tau},D}(\mathbf n,\mathbf n')
=
\ell_{\boldsymbol{\tau},D}(\mathbf n,\mathbf n')
-
u_{\boldsymbol{\tau},D}(\mathbf n,\mathbf n').
\end{equation}
With this notation,
$\mathcal J_D(\mathbf n,\mathbf n')
\cap
\mathcal W_{\boldsymbol{\tau}}(\mathbf n,\mathbf n')$
is empty if and only if
\begin{equation}\label{eq:pairwise_positive_condition}
\Theta_{\boldsymbol{\tau},D}(\mathbf n,\mathbf n')>0.
\end{equation}

The robust recovery requires this separation for every pair of distinct folding integer vectors.
We therefore define the global separation margin by
\begin{equation}\label{eq:Delta_full_def}
\Delta(\boldsymbol{\tau},D)
=
\begin{cases}
+\infty,
&
|\mathcal F(D)|=1,\\[1mm]
\displaystyle
\min_{\substack{
\mathbf n,\mathbf n'\in\mathcal F(D)\\
\mathbf n\ne\mathbf n'
}}
\Theta_{\boldsymbol{\tau},D}(\mathbf n,\mathbf n'),
&
|\mathcal F(D)|\ge2.
\end{cases}
\end{equation}
The case $|\mathcal F(D)|=1$ means that all integers in
$[0,D-1]_{\mathbb Z}$ have the same folding integer vector, so there is no pair of
distinct folding integer vectors to separate.

\begin{example}\label{ex:theta_delta_example}
Consider again Example~\ref{ex:folding_vector_example}. For short, write
$\mathbf n_0=(0,0)$, $\mathbf n_1=(1,0)$, and $\mathbf n_2=(1,1)$.

Let $\boldsymbol{\tau}=(0,1)$. We compute the separation margin for
$\mathbf n_0$ and $\mathbf n_2$. Since
$\mathcal I_{12}(\mathbf n_0)=[0,5]_{\mathbb Z}$ and
$\mathcal I_{12}(\mathbf n_2)=[10,11]_{\mathbb Z}$, we have
\[
\mathcal J_{12}(\mathbf n_0,\mathbf n_2)
=
\mathcal I_{12}(\mathbf n_0)-\mathcal I_{12}(\mathbf n_2)
=
[-11,-5]_{\mathbb Z}.
\]
Also, $c_1(\mathbf n_0,\mathbf n_2)=-6$ and
$c_2(\mathbf n_0,\mathbf n_2)=-10$. Therefore,
$p_{\boldsymbol{\tau}}(\mathbf n_0,\mathbf n_2)=-6$ and
$q_{\boldsymbol{\tau}}(\mathbf n_0,\mathbf n_2)=-8$. Hence
$p_{\boldsymbol{\tau}}(\mathbf n_0,\mathbf n_2)>
q_{\boldsymbol{\tau}}(\mathbf n_0,\mathbf n_2)$, so
$\mathcal W_{\boldsymbol{\tau}}(\mathbf n_0,\mathbf n_2)=\emptyset$.
Equivalently,
\[
\ell_{\boldsymbol{\tau},12}(\mathbf n_0,\mathbf n_2)=-6,
\qquad
u_{\boldsymbol{\tau},12}(\mathbf n_0,\mathbf n_2)=-8,
\]
and therefore
\[
\Theta_{\boldsymbol{\tau},12}(\mathbf n_0,\mathbf n_2)
=
-6-(-8)=2>0.
\]
For the other two pairs, one similarly obtains
$\Theta_{\boldsymbol{\tau},12}(\mathbf n_0,\mathbf n_1)=4$ and
$\Theta_{\boldsymbol{\tau},12}(\mathbf n_1,\mathbf n_2)=8$. Thus,
\[
\Delta(\boldsymbol{\tau},12)=\min\{4,2,8\}=2.
\]
Therefore, all distinct folding integer vectors in $\mathcal F(12)$ are
separated under $\boldsymbol{\tau}=(0,1)$. Equivalently, no admissible erroneous
remainder vector can be generated by two different folding integer vectors, and
robust determination holds for $D=12$ according to
Lemma~\ref{lem:observation_set_criterion}.

If the error bound vector is changed to $\boldsymbol{\tau}=(1,1)$, then for the
same pair $\mathbf n_0$ and $\mathbf n_2$, we obtain
$p_{\boldsymbol{\tau}}(\mathbf n_0,\mathbf n_2)
=
q_{\boldsymbol{\tau}}(\mathbf n_0,\mathbf n_2)
=
-8$. Thus
$\mathcal W_{\boldsymbol{\tau}}(\mathbf n_0,\mathbf n_2)=\{-8\}$. Since
$-8\in\mathcal J_{12}(\mathbf n_0,\mathbf n_2)$, the two intervals now
intersect, and
\[
\Theta_{\boldsymbol{\tau},12}(\mathbf n_0,\mathbf n_2)=0.
\]
Therefore, one admissible erroneous remainder vector may correspond to two
different folding integer vectors, and robust determination no longer holds for
$D=12$ under $\boldsymbol{\tau}=(1,1)$.
\end{example}

The following theorem gives the exact robust determination condition for fixed $D$ and $\boldsymbol{\tau}$.

\begin{thm}\label{thm:exact_fixed_D_tradeoff}
Fix a range length $1\le D\le M$ and an error bound vector
$\boldsymbol{\tau}\in\mathbb Z_{\ge0}^L$. Under the remainder error model
\eqref{eq:error_model}--\eqref{eq:nonuniform_error_bound}, every integer
$N\in[0,D-1]_{\mathbb Z}$ can be robustly determined if and only if
\begin{equation}\label{eq:exact_full_condition}
\Delta(\boldsymbol{\tau},D)>0.
\end{equation}
\end{thm}

\begin{proof}
This is a direct consequence of Lemma~\ref{lem:observation_set_criterion},
Lemma~\ref{lem:interval_overlap}, and
\eqref{eq:pairwise_positive_condition}--\eqref{eq:Delta_full_def}, where the
minimization is taken only over the finite set $\mathcal F(D)$.
\end{proof}

Theorem~\ref{thm:exact_fixed_D_tradeoff} reduces the robust determination for
fixed $D$ and $\boldsymbol{\tau}$ to the sign of
$\Delta(\boldsymbol{\tau},D)$. We next present two monotonicity properties
that will be used to characterize the dynamic range
$D(\boldsymbol{\tau})$ for a fixed remainder error bound vector $\boldsymbol{\tau}$ and the admissible remainder error
bounds for a fixed range length.

\begin{thm}\label{thm:margin_monotonicity}
The separation margin $\Delta(\boldsymbol{\tau},D)$ in \eqref{eq:Delta_full_def} satisfies the following monotonicity properties.

(i) For a fixed remainder error bound vector $\boldsymbol{\tau}$, if $1\le D_1\le D_2\le M$, then
\begin{equation}\label{eq:monotone_D}
\Delta(\boldsymbol{\tau},D_2)
\le
\Delta(\boldsymbol{\tau},D_1).
\end{equation}

(ii) For a fixed range length $D$, if $\boldsymbol{\tau}\preceq\boldsymbol{\eta}$, then
\begin{equation}\label{eq:monotone_tau}
\Delta(\boldsymbol{\eta},D)
\le
\Delta(\boldsymbol{\tau},D).
\end{equation}
\end{thm}

\begin{proof}
We first prove \eqref{eq:monotone_D}. Suppose $1\le D_1\le D_2\le M$.
Then
\[
[0,D_1-1]_{\mathbb Z}
\subseteq
[0,D_2-1]_{\mathbb Z}.
\]
Hence every folding integer vector in $\mathcal F(D_1)$ also belongs to
$\mathcal F(D_2)$. Fix two distinct folding integer vectors
$\mathbf n,\mathbf n'\in\mathcal F(D_1)$. Since
$\mathcal I_{D_1}(\mathbf n)
\subseteq
\mathcal I_{D_2}(\mathbf n)$ and
$\mathcal I_{D_1}(\mathbf n')
\subseteq
\mathcal I_{D_2}(\mathbf n')$,
we obtain
\[
\mathcal J_{D_1}(\mathbf n,\mathbf n')
\subseteq
\mathcal J_{D_2}(\mathbf n,\mathbf n').
\]
The interval $\mathcal W_{\boldsymbol{\tau}}(\mathbf n,\mathbf n')$ is
independent of $D$. Therefore, for the same pair
$\mathbf n,\mathbf n'$, the separation margin cannot increase, that is
\[
\Theta_{\boldsymbol{\tau},D_2}(\mathbf n,\mathbf n')
\le
\Theta_{\boldsymbol{\tau},D_1}(\mathbf n,\mathbf n').
\]
Moreover, since $\mathcal F(D_1)\subseteq \mathcal F(D_2)$, the minimum
defining $\Delta(\boldsymbol{\tau},D_2)$ in \eqref{eq:Delta_full_def} is taken over a set of pairs that
contains all pairs from $\mathcal F(D_1)$. Thus,
\[
\Delta(\boldsymbol{\tau},D_2)
\le
\Delta(\boldsymbol{\tau},D_1).
\]

We next prove \eqref{eq:monotone_tau}. Suppose
$\boldsymbol{\tau}\preceq\boldsymbol{\eta}$. Then, for every pair
$\mathbf n\ne\mathbf n'$,
\[
\mathcal W_{\boldsymbol{\tau}}(\mathbf n,\mathbf n')
\subseteq
\mathcal W_{\boldsymbol{\eta}}(\mathbf n,\mathbf n').
\]
The interval $\mathcal J_D(\mathbf n,\mathbf n')$ is not related to remainders. Thus, for
each pair, enlarging the remainder error bound vector cannot increase the separation
margin, that is 
\[
\Theta_{\boldsymbol{\eta},D}(\mathbf n,\mathbf n')
\le
\Theta_{\boldsymbol{\tau},D}(\mathbf n,\mathbf n').
\]
Taking the minimum over all distinct pairs in $\mathcal F(D)$ gives
\[
\Delta(\boldsymbol{\eta},D)
\le
\Delta(\boldsymbol{\tau},D). \qedhere
\]
\end{proof}

\subsection{Dynamic Range for a Given Error Bound}
\label{s3_sub4}

We now fix a remainder error bound vector $\boldsymbol{\tau}$ and characterize the
dynamic range $D(\boldsymbol{\tau})$. By
Theorem~\ref{thm:exact_fixed_D_tradeoff}, a range length $D$ is admissible
under $\boldsymbol{\tau}$ if and only if
$\Delta(\boldsymbol{\tau},D)>0$. By
Theorem~\ref{thm:margin_monotonicity}, once this condition fails for
some $D$, it fails for every larger range length. Hence the largest admissible
range length is exactly
\begin{equation}\label{eq:D_tau_by_margin}
D(\boldsymbol{\tau})
=
\max
\left\{
D\in[1,M]_{\mathbb Z}:
\Delta(\boldsymbol{\tau},D)>0
\right\}.
\end{equation}

\begin{remark}\label{rem:D_tau_computation}
Equation~\eqref{eq:D_tau_by_margin} already gives an exact way to obtain
$D(\boldsymbol{\tau})$. Since the condition
$\Delta(\boldsymbol{\tau},D)>0$ is monotonic in $D$, one can find the largest
admissible range length by a one-dimensional search over $[1,M]_{\mathbb Z}$.
The next result gives a more structural description. It shows that
$D(\boldsymbol{\tau})$ is determined by the earliest integer point where ambiguity first appears as
the range length increases. Here, an ambiguity means that two integers with
different folding integer vectors can generate the same erroneous remainder vector
under the error bound vector $\boldsymbol{\tau}$.
\end{remark}

To describe the earliest integer point with ambiguity, we consider all folding integer vectors over
the largest error-free CRT dynamic range $[0,M-1]_{\mathbb Z}$.
Fix two distinct folding integer vectors
$\mathbf n,\mathbf n'\in\mathcal F(M)$, and keep the order $(\mathbf n,\mathbf n')$ fixed for this pair
$\mathbf n$ and $\mathbf n'$ in the definitions below.  Then
\[
\mathcal I_M(\mathbf n)
=
[\alpha_M(\mathbf n),\beta_M(\mathbf n)]_{\mathbb Z},
\,
\mathcal I_M(\mathbf n')
=
[\alpha_M(\mathbf n'),\beta_M(\mathbf n')]_{\mathbb Z}.
\]
By Lemma~\ref{lem:interval_overlap}, this pair $(\mathbf n,\mathbf n')$ can produce an
ambiguity under the remainder error bound vector $\boldsymbol{\tau}$ if and only if
\[
\mathcal J_M(\mathbf n,\mathbf n')
\cap
\mathcal W_{\boldsymbol{\tau}}(\mathbf n,\mathbf n')
\ne
\emptyset.
\]
By \eqref{eq:JW_as_interval}, the feasible values of $s=N-N'$ that can
produce such an ambiguity are exactly
\begin{equation}\label{eq:feasible_s_interval_M}
[
\ell_{\boldsymbol{\tau},M}(\mathbf n,\mathbf n'),
u_{\boldsymbol{\tau},M}(\mathbf n,\mathbf n')
]_{\mathbb Z}.
\end{equation}
If this interval is empty, then this pair $(\mathbf n,\mathbf n')$ cannot produce an ambiguity in $[0,M-1]_{\mathbb Z}$ under $\boldsymbol{\tau}$.

We now characterize where the first ambiguity from a fixed pair $(\mathbf n,\mathbf n')$ can
appear. Suppose that the interval in \eqref{eq:feasible_s_interval_M} is
nonempty. Since $\mathbf n\ne\mathbf n'$, the sets
$\mathcal I_M(\mathbf n)$ and $\mathcal I_M(\mathbf n')$ are disjoint. Hence
no feasible value of $s=N-N'$ can be zero. Therefore, the feasible interval in
\eqref{eq:feasible_s_interval_M} lies either entirely on the positive side or
entirely on the negative side.

Define
\begin{equation}\label{eq:s_star_def}
s_{\boldsymbol{\tau}}^\star(\mathbf n,\mathbf n')
=
\begin{cases}
\ell_{\boldsymbol{\tau},M}(\mathbf n,\mathbf n'),
&
\ell_{\boldsymbol{\tau},M}(\mathbf n,\mathbf n')>0,\\[1mm]
u_{\boldsymbol{\tau},M}(\mathbf n,\mathbf n'),
&
u_{\boldsymbol{\tau},M}(\mathbf n,\mathbf n')<0.
\end{cases}
\end{equation}
Thus, $s_{\boldsymbol{\tau}}^\star(\mathbf n,\mathbf n')$ is the feasible
value of $s$ closest to zero. For this pair, define
\begin{equation}\label{eq:T_tau_def}
\begin{aligned}
T_{\boldsymbol{\tau}}(\mathbf n,\mathbf n')
=
&\max\{
\alpha_M(\mathbf n'),
\alpha_M(\mathbf n)-s_{\boldsymbol{\tau}}^\star(\mathbf n,\mathbf n')
\}  \\
&+\max\{
s_{\boldsymbol{\tau}}^\star(\mathbf n,\mathbf n'),0
\}.
\end{aligned}
\end{equation}

\begin{lemma}
\label{lem:first_ambiguity_pair}
Fix two distinct folding integer vectors
$\mathbf n,\mathbf n'\in\mathcal F(M)$. Suppose that
\[
\mathcal J_M(\mathbf n,\mathbf n')
\cap
\mathcal W_{\boldsymbol{\tau}}(\mathbf n,\mathbf n')
\ne
\emptyset.
\]
Then $T_{\boldsymbol{\tau}}(\mathbf n,\mathbf n')$ is the smallest possible
value of $\max\{N,N'\}$ over all integer pairs $N,N'$ satisfying
\[
N\in\mathcal I_M(\mathbf n),
\qquad
N'\in\mathcal I_M(\mathbf n'),
\]
such that $N$ and $N'$ can generate the same erroneous remainder vector under
the remainder error bound vector $\boldsymbol{\tau}$.
\end{lemma}

\begin{proof}
Let $s=N-N'$. By Lemma~\ref{lem:interval_overlap}, the feasible values of
$s$ that can produce an ambiguity for this pair $(\mathbf n,\mathbf n')$ are exactly the
interval in \eqref{eq:feasible_s_interval_M}.

Fix one feasible value of $s$. Write $N'=y$ and $N=y+s$. The constraints
$N'\in\mathcal I_M(\mathbf n')$ and $N\in\mathcal I_M(\mathbf n)$ become
\[
y\in\mathcal I_M(\mathbf n'),
\qquad
y+s\in\mathcal I_M(\mathbf n).
\]
Since
\[
\mathcal I_M(\mathbf n)
=
[\alpha_M(\mathbf n),\beta_M(\mathbf n)]_{\mathbb Z},
\,
\mathcal I_M(\mathbf n')
=
[\alpha_M(\mathbf n'),\beta_M(\mathbf n')]_{\mathbb Z},
\]
the smallest feasible value of $y$ is
\[
y_{\min}(s)
=
\max\{\alpha_M(\mathbf n'),\alpha_M(\mathbf n)-s\}.
\]
Choosing $y=y_{\min}(s)$, or equivalently
$N'=y_{\min}(s)$ and $N=y_{\min}(s)+s$, gives
\[
\begin{aligned}
\max\{N,N'\}
&=
\max\{y_{\min}(s)+s,\ y_{\min}(s)\} \\
&=
y_{\min}(s)+\max\{s,0\}.
\end{aligned}
\]
Therefore, for a fixed feasible $s$, the smallest possible value of
$\max\{N,N'\}$ is
\begin{equation}\label{eq:minmax_for_fixed_s}
\max\{\alpha_M(\mathbf n'),\alpha_M(\mathbf n)-s\}
+
\max\{s,0\}.
\end{equation}
The quantity in \eqref{eq:minmax_for_fixed_s} is nonincreasing as a negative
feasible $s$ moves toward zero and is nondecreasing as a positive feasible
$s$ moves away from zero. Hence
the minimum over the feasible interval in \eqref{eq:feasible_s_interval_M} is
attained at the feasible value closest to zero, namely
$s_{\boldsymbol{\tau}}^\star(\mathbf n,\mathbf n')$. Substituting this value
into \eqref{eq:minmax_for_fixed_s} gives \eqref{eq:T_tau_def}. This proves
the lemma.
\end{proof}

Define
\begin{equation}\label{eq:conflicting_ordered_pair_set}
\mathcal P_{\boldsymbol{\tau}}
=
\left\{
(\mathbf n,\mathbf n'):
\begin{array}{l}
\mathbf n,\mathbf n'\in\mathcal F(M),\quad
\mathbf n\ne\mathbf n',\\
\ell_{\boldsymbol{\tau},M}(\mathbf n,\mathbf n')
\le
u_{\boldsymbol{\tau},M}(\mathbf n,\mathbf n')
\end{array}
\right\}.
\end{equation}
Thus, $\mathcal P_{\boldsymbol{\tau}}$ is the set of pairs $(\mathbf n,\mathbf n')$ of folding integer vectors that can generate an ambiguity under $\boldsymbol{\tau}$ within
$[0,M-1]_{\mathbb Z}$.

\begin{thm}[Dynamic range for a given remainder error bound vector]
\label{thm:exact_D_tau}
For any remainder error bound vector $\boldsymbol{\tau}$,
\begin{equation}\label{eq:D_tau_earliest_conflict}
D(\boldsymbol{\tau})
=
\begin{cases}
M,
&
\mathcal P_{\boldsymbol{\tau}}=\emptyset,\\[1mm]
\displaystyle
\min_{(\mathbf n,\mathbf n')\in\mathcal P_{\boldsymbol{\tau}}}
T_{\boldsymbol{\tau}}(\mathbf n,\mathbf n'),
&
\mathcal P_{\boldsymbol{\tau}}\ne\emptyset.
\end{cases}
\end{equation}
\end{thm}

\begin{proof}
If $\mathcal P_{\boldsymbol{\tau}}=\emptyset$, then no two distinct folding
integer vectors can produce overlapping observation sets within
$[0,M-1]_{\mathbb Z}$. By Lemma~\ref{lem:observation_set_criterion}, every
integer in $[0,M-1]_{\mathbb Z}$ can be robustly determined. Since no range
length larger than $M$ is possible even in the error-free case, we have
$D(\boldsymbol{\tau})=M$.

Now suppose that $\mathcal P_{\boldsymbol{\tau}}\ne\emptyset$, and define
\[
K
=
\min_{(\mathbf n,\mathbf n')\in\mathcal P_{\boldsymbol{\tau}}}
T_{\boldsymbol{\tau}}(\mathbf n,\mathbf n').
\]
By Lemma~\ref{lem:first_ambiguity_pair}, $K$ is the smallest possible value
of $\max\{N,N'\}$ over all pairs of integers in $[0,M-1]_{\mathbb Z}$ with
different folding integer vectors that can generate the same erroneous
remainder vector under $\boldsymbol{\tau}$. Hence no such pair is contained
in $[0,K-1]_{\mathbb Z}$, while at least one such pair is contained in
$[0,K]_{\mathbb Z}$. Therefore the robust determination holds for range length
$K$ and fails for range length $K+1$. By the monotonicity in
Theorem~\ref{thm:margin_monotonicity}, the dynamic range is
$D(\boldsymbol{\tau})=K$, which proves \eqref{eq:D_tau_earliest_conflict}.
\end{proof}

The following example shows how different remainder error bound vectors can lead to different dynamic ranges.

\begin{example}
Let $(m_1,m_2,m_3)=(6,9,15)$. Then
\[
m=\gcd(6,9,15)=3,
\qquad
M=\lcm(6,9,15)=90.
\]
Using Theorem~\ref{thm:exact_D_tau}, we obtain the following dynamic ranges $D(\boldsymbol{\tau})$ for different remainder error bound vectors $\boldsymbol{\tau}$:
\[
\begin{aligned}
D(1,1,1)&=18,
&
D(1,1,2)&=15,
&
D(1,2,1)&=15,\\
D(2,1,1)&=6,
&
D(2,2,2)&=6.
\end{aligned}
\]

This example illustrates how the dynamic range changes with the remainder error bound
vector. First, the two non-uniform bounds $(1,1,2)$ and $(1,2,1)$ lie between
the two uniform bounds $(1,1,1)$ and $(2,2,2)$.
Their dynamic ranges are also between the two corresponding uniform dynamic
ranges:
\[
D(2,2,2)
<
D(1,1,2)
=
D(1,2,1)
<
D(1,1,1).
\]
Thus, a non-uniform error bound can give an intermediate dynamic range that
is not captured by using only one uniform error level.

Second, individual remainder error bounds for different moduli matter. The three remainder error bound vectors
$(1,1,2)$, $(1,2,1)$, and $(2,1,1)$ have the same maximum component, but they
do not give the same dynamic range.
Therefore, $D(\boldsymbol{\tau})$ depends on the full error bound vector
$\boldsymbol{\tau}$, not only on $\max_i\tau_i$.
\end{example}

\subsection{Tolerable Remainder Error Bounds for a Given Range}\label{s3_sub5}

We now fix a range length $D$ and characterize the remainder error bound vectors that
can be tolerated on $[0,D-1]_{\mathbb Z}$. By
Theorem~\ref{thm:exact_fixed_D_tradeoff}, a vector $\boldsymbol{\tau}$ is
tolerable for this range if and only if
$\Delta(\boldsymbol{\tau},D)>0$. Define
\begin{equation}\label{eq:tolerable_region_def}
\mathcal T(D)
=
\left\{
\boldsymbol{\tau}\in\mathbb Z_{\ge0}^L:
\Delta(\boldsymbol{\tau},D)>0
\right\}.
\end{equation}
Thus, $\mathcal T(D)$ is the set of all remainder error bound vectors under which every
integer in $[0,D-1]_{\mathbb Z}$ can be robustly determined. By Theorem~\ref{thm:margin_monotonicity}, if
$\boldsymbol{\tau}\in\mathcal T(D)$ and
$\boldsymbol{\eta}\preceq\boldsymbol{\tau}$, then
$\boldsymbol{\eta}\in\mathcal T(D)$.

To describe $\mathcal T(D)$, we examine one possible way in which robust
determination can fail. Fix two distinct folding integer vectors
$\mathbf n,\mathbf n'\in\mathcal F(D)$ and a value
$s\in\mathcal J_D(\mathbf n,\mathbf n')$. This means that there exist
integers
\[
N\in\mathcal I_D(\mathbf n),
\qquad
N'\in\mathcal I_D(\mathbf n')
\]
such that $s=N-N'\neq 0$. For a given remainder error bound vector $\boldsymbol{\tau}$, this
triple $(\mathbf n,\mathbf n',s)$ can make the two observation sets overlap
if and only if
\begin{equation}\label{eq:s_overlap_condition_tau}
|s-c_i(\mathbf n,\mathbf n')|
\le
2\tau_i,
\qquad
i=1,\ldots,L,
\end{equation}
which is exactly the condition that
$s\in\mathcal W_{\boldsymbol{\tau}}(\mathbf n,\mathbf n')$.

For this fixed triple $(\mathbf n,\mathbf n',s)$, define
\begin{equation}\label{eq:lambda_forbidden_def}
\boldsymbol{\lambda}_{\mathbf n,\mathbf n',s}
=
(\lambda_1,\ldots,\lambda_L),
\qquad
\lambda_i
=
\left\lceil
\frac{|s-c_i(\mathbf n,\mathbf n')|}{2}
\right\rceil.
\end{equation}
Since $\tau_i$ is an integer, \eqref{eq:s_overlap_condition_tau} holds if and
only if
\begin{equation}\label{eq:lambda_dominated_condition}
\boldsymbol{\lambda}_{\mathbf n,\mathbf n',s}
\preceq
\boldsymbol{\tau}.
\end{equation}
Thus, for the fixed triple $(\mathbf n,\mathbf n',s)$, the error bound
vectors $\boldsymbol{\tau}$ for which the corresponding observation sets
$\mathcal O_{\boldsymbol{\tau},D}(\mathbf n)$ and
$\mathcal O_{\boldsymbol{\tau},D}(\mathbf n')$ overlap through the difference
$s$ are exactly those satisfying
$\boldsymbol{\lambda}_{\mathbf n,\mathbf n',s}\preceq\boldsymbol{\tau}$.

Collect all vectors of the form
$\boldsymbol{\lambda}_{\mathbf n,\mathbf n',s}$ in the finite set
\begin{equation}\label{eq:Lambda_D_def}
\Lambda_D
=
\left\{
\boldsymbol{\lambda}_{\mathbf n,\mathbf n',s}:
\mathbf n,\mathbf n'\in\mathcal F(D),\
\mathbf n\ne\mathbf n',\
s\in\mathcal J_D(\mathbf n,\mathbf n')
\right\}.
\end{equation}
This set may contain redundant vectors. We therefore keep only its minimal
elements under the componentwise order and define
\begin{equation}\label{eq:Lambda_min_def}
\Lambda_D^{\min}
=
\left\{
\boldsymbol{\lambda}\in\Lambda_D:
\text{there is no }
\boldsymbol{\mu}\in\Lambda_D
\text{ such that }
\boldsymbol{\mu}\prec\boldsymbol{\lambda}
\right\},
\end{equation}
where $\boldsymbol{\mu}\prec\boldsymbol{\lambda}$ means
$\boldsymbol{\mu}\preceq\boldsymbol{\lambda}$ and
$\boldsymbol{\mu}\ne\boldsymbol{\lambda}$.
Thus, no two distinct elements of $\Lambda_D^{\min}$ are comparable under
$\preceq$.

For any set $\Lambda\subseteq\mathbb Z_{\ge0}^L$, define
\begin{equation}\label{eq:upper_closure_def}
\mathcal U(\Lambda)
=
\left\{
\boldsymbol{\tau}\in\mathbb Z_{\ge0}^L:
\boldsymbol{\lambda}\preceq\boldsymbol{\tau}
\text{ for some }
\boldsymbol{\lambda}\in\Lambda
\right\}.
\end{equation}
Then, $\mathcal U(\Lambda_D^{\min})$ contains exactly those remainder error bound vectors
that dominate at least one vector in $\Lambda_D^{\min}$.
The set $\mathcal U(\Lambda_D^{\min})$ gives the following characterization of $\mathcal T(D)$.

\begin{thm}[Exact characterization of tolerable remainder error bound vectors]
\label{thm:tolerable_error_bounds}
For a fixed $D$ with $1\le D\le M$,
\begin{equation}\label{eq:tolerable_region_exact}
\mathcal T(D)
=
\mathbb Z_{\ge0}^L
\setminus
\mathcal U(\Lambda_D^{\min}).
\end{equation}
Equivalently,
\begin{equation}\label{eq:tolerable_region_equiv}
\mathcal T(D)
=
\left\{
\boldsymbol{\tau}\in\mathbb Z_{\ge0}^L:
\boldsymbol{\lambda}\npreceq\boldsymbol{\tau}
\text{ for every }
\boldsymbol{\lambda}\in\Lambda_D^{\min}
\right\}.
\end{equation}
\end{thm}

\begin{proof}
We first characterize when a remainder error bound vector is not tolerable. By
Lemma~\ref{lem:interval_overlap}, $\boldsymbol{\tau}\notin\mathcal T(D)$ if
and only if there exist two distinct folding integer vectors
$\mathbf n,\mathbf n'\in\mathcal F(D)$ and a non-zero value
$s\in\mathcal J_D(\mathbf n,\mathbf n')$ such that
$s\in\mathcal W_{\boldsymbol{\tau}}(\mathbf n,\mathbf n')$. By
\eqref{eq:lambda_dominated_condition}, this is equivalent to the existence
of some $\boldsymbol{\lambda}_{\mathbf n,\mathbf n',s}\in\Lambda_D$ such that
$\boldsymbol{\lambda}_{\mathbf n,\mathbf n',s}\preceq\boldsymbol{\tau}$.
Thus, $\boldsymbol{\tau}$ is not tolerable if and only if it dominates at
least one vector in $\Lambda_D$.

It remains to remove the redundant vectors in $\Lambda_D$. Since
$\Lambda_D$ is finite, every vector in $\Lambda_D$ dominates at least one
vector in $\Lambda_D^{\min}$. Therefore, $\boldsymbol{\tau}$ dominates a
vector in $\Lambda_D$ if and only if it dominates a vector in
$\Lambda_D^{\min}$. Hence the non-tolerable error bound vectors are exactly
the vectors in $\mathcal U(\Lambda_D^{\min})$. Taking the complement in
$\mathbb Z_{\ge0}^L$ gives
\[
\mathcal T(D)
=
\mathbb Z_{\ge0}^L
\setminus
\mathcal U(\Lambda_D^{\min}).
\]
The equivalent form in \eqref{eq:tolerable_region_equiv} follows directly
from the definition of $\mathcal U(\Lambda_D^{\min})$.
\end{proof}

Theorem~\ref{thm:tolerable_error_bounds} characterizes the whole set
$\mathcal T(D)$ by the finite boundary set $\Lambda_D^{\min}$. To test whether a remainder error bound vector
$\boldsymbol{\tau}$ can be tolerated on $[0,D-1]_{\mathbb Z}$, it is enough
to check whether it dominates any vector in the finite set
$\Lambda_D^{\min}$. If it dominates any element in $\Lambda_D^{\min}$, then it causes an ambiguity;
otherwise, it can be tolerated.

The following example illustrates how Theorem~\ref{thm:tolerable_error_bounds}
describes the remainder error bound vectors that can be tolerated for a fixed range
length.

\begin{example}\label{ex:tolerable_region_example}
Let $(m_1,m_2,m_3)=(6,10,15)$. Then $M=\lcm(6,10,15)=30$. Fix the range length $D=12$. The folding integer vectors of the integers in $[0,11]_{\mathbb Z}$ are
\[
\mathcal F(12)=\{(0,0,0),(1,0,0),(1,1,0)\}.
\]
The corresponding integer sets are $\mathcal I_{12}(0,0,0)=[0,5]_{\mathbb Z}$, $\mathcal I_{12}(1,0,0)=[6,9]_{\mathbb Z}$, and $\mathcal I_{12}(1,1,0)=[10,11]_{\mathbb Z}$. Using \eqref{eq:lambda_forbidden_def} and \eqref{eq:Lambda_min_def}, we obtain
\begin{equation}\label{eq:example_Lambda_min}
\begin{aligned}
\Lambda_{12}^{\min}
=
\{&
(0,2,3),\,
(1,1,4),\,
(1,2,2),\\
&
(1,4,1),\,
(2,0,5),\,
(2,1,1)
\}.
\end{aligned}
\end{equation}
By Theorem~\ref{thm:tolerable_error_bounds}, a remainder error bound vector $\boldsymbol{\tau}=(\tau_1,\tau_2,\tau_3)$ can be tolerated on $[0,11]_{\mathbb Z}$ if and only if it does not dominate any vector in $\Lambda_{12}^{\min}$. Equivalently,
\begin{equation}\label{eq:example_T_12}
\mathcal T(12)
=
\left\{
\boldsymbol{\tau}\in\mathbb Z_{\ge0}^3:
\boldsymbol{\lambda}\npreceq\boldsymbol{\tau}
\text{ for every }
\boldsymbol{\lambda}\in\Lambda_{12}^{\min}
\right\}.
\end{equation}

Fig.~\ref{fig:tolerable_region_slices} shows slices of $\mathcal T(12)$ inside the box $[0,5]_{\mathbb Z}^3$. Each panel fixes $\tau_3$ and displays the tolerable and not tolerable points in the $(\tau_1,\tau_2)$ plane. The figure gives a direct visual description of the set in \eqref{eq:example_T_12}.

The set in \eqref{eq:example_Lambda_min} also gives a simple way to check individual vectors. For instance, $(0,1,5)$ and $(1,0,5)$ belong to $\mathcal T(12)$ because neither of them dominates any vector in $\Lambda_{12}^{\min}$. However, $(1,1,5)\notin\mathcal T(12)$ because it dominates $(1,1,4)\in\Lambda_{12}^{\min}$. Therefore, the tolerable error bound vectors cannot be described by independent upper bounds on $\tau_1,\tau_2,\tau_3$. Although $(0,1,5)$ and $(1,0,5)$ can both be tolerated, increasing both the first and second components to obtain $(1,1,5)$ makes the vector not tolerable.

This example also shows that the largest component of $\boldsymbol{\tau}$ alone is not enough to determine the error tolerability. The remainder error bound vectors $(0,1,5)$, $(1,0,5)$, and $(1,1,5)$ all have maximum component $5$, but only the first two can be tolerated. In contrast, the uniform vector $(2,2,2)$ is not tolerable, since it dominates $(1,2,2)\in\Lambda_{12}^{\min}$. Hence the largest tolerable uniform integer remainder error bound for $D=12$ is $1$, while the non-uniform remainder error bound vector $(0,1,5)$ can still be tolerated. Thus, for a fixed range length, the full remainder error bound vector $\boldsymbol{\tau}$ matters. A scalar metric such as $\max_i\tau_i$ is not enough.
\end{example}

\begin{figure}[htbp]
\centering
\includegraphics[width=0.4\textwidth]{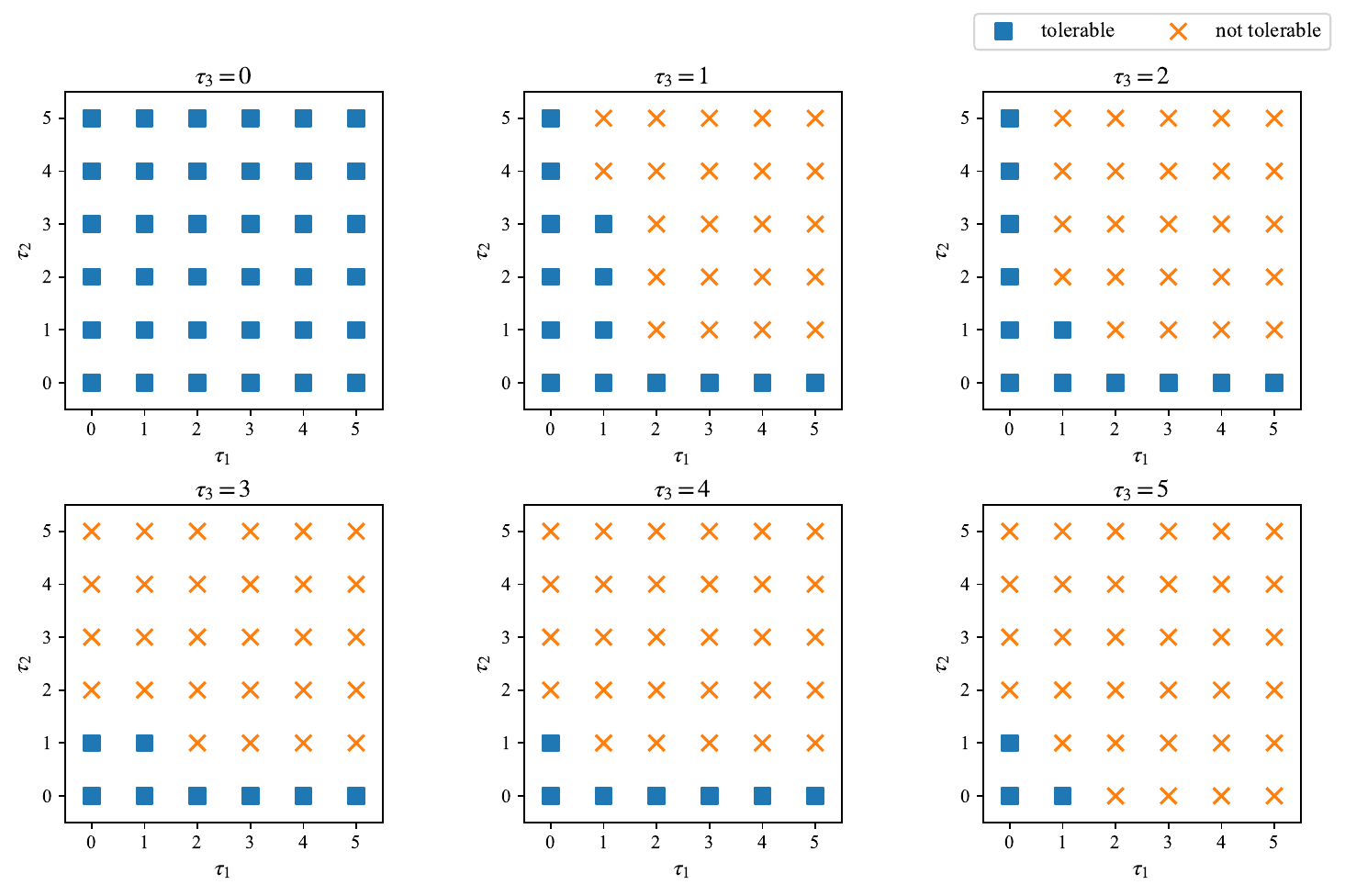}
\caption{Slices of the tolerable remainder error bound vector set $\mathcal T(12)$ for $(m_1,m_2,m_3)=(6,10,15)$ inside $[0,5]_{\mathbb Z}^3$. Each panel fixes $\tau_3$ and displays the tolerable and not tolerable points in the $(\tau_1,\tau_2)$ plane.}
\label{fig:tolerable_region_slices}
\end{figure}

\section{Decoding Algorithm}
\label{s4}

In this section, we present a decoding algorithm for the folding integer vector. 
The range length $D$ and the remainder error bound vector $\boldsymbol{\tau}$ are assumed to be given and satisfy $\Delta(\boldsymbol{\tau},D)>0$.

The key idea is the following. Given an erroneous remainder vector
$\widetilde{\mathbf r}$, the algorithm searches for an integer
$a\in[0,D-1]_{\mathbb Z}$ that is feasible for this observation $\widetilde{\mathbf r}$ under the
remainder error bound vector $\boldsymbol{\tau}$. Here, ``feasible'' means that
there exists an error vector
$\mathbf e^{(a)}=(e_1^{(a)},\ldots,e_L^{(a)})\in\mathcal E_{\boldsymbol{\tau}}$
such that
\[
\widetilde r_i=(a\bmod m_i)+e_i^{(a)},
\qquad i=1,\ldots,L.
\]
Equivalently,
\[
|\widetilde r_i-(a\bmod m_i)|\le \tau_i,
\qquad i=1,\ldots,L.
\]
A feasible integer may not be unique. However, when
$\Delta(\boldsymbol{\tau},D)>0$, all feasible integers have the same folding
integer vector. Therefore, once the algorithm finds one feasible integer, the
folding integer vector of this integer is the desired output.

For each erroneous remainder $\widetilde r_i$, define the set 
\begin{equation}\label{eq:possible_true_remainder_set}
\mathcal R_i
=
[\widetilde r_i-\tau_i,\widetilde r_i+\tau_i]_{\mathbb Z}
\cap
[0,m_i-1]_{\mathbb Z},
\qquad i=1,\ldots,L.
\end{equation}
Then an integer $a\in[0,D-1]_{\mathbb Z}$ is feasible for
$\widetilde{\mathbf r}$ under the remainder error bound vector $\boldsymbol{\tau}$ if and only if
\begin{equation}\label{eq:integer_candidate_constraints}
a\bmod m_i\in\mathcal R_i,
\qquad \text{for all } i=1,\ldots,L.
\end{equation}
Therefore, to find such an integer, we may choose one 
from each set $\mathcal R_i$ and check whether these chosen remainders can
come from a common integer in $[0,D-1]_{\mathbb Z}$. More precisely, for a tuple
\[
(\rho_1,\ldots,\rho_L)
\in
\mathcal R_1\times\cdots\times\mathcal R_L,
\]
we consider the congruence system
\begin{equation}\label{eq:chosen_remainder_crt_system}
x\equiv \rho_i \mod {m_i},
\qquad i=1,\ldots,L.
\end{equation}
If this system has a solution, then its solution is unique modulo
$M=\lcm(m_1,\ldots,m_L)$. Let $a\in[0,M-1]_{\mathbb Z}$ be the unique
solution. If $a<D$, then $a$ is
feasible for the observation $\widetilde{\mathbf r}$ under
$\boldsymbol{\tau}$, or equivalently,
$\widetilde{\mathbf r}
\in
\mathcal O_{\boldsymbol{\tau},D}(\mathbf n(a))$.
If the true integer is $N \in[0,D-1]_{\mathbb Z}$, then also
$\widetilde{\mathbf r}
\in
\mathcal O_{\boldsymbol{\tau},D}(\mathbf n(N))$.
Since $\Delta(\boldsymbol{\tau},D)>0$, observation sets corresponding to
distinct folding integer vectors are disjoint. Therefore,
$\mathbf n(a)=\mathbf n(N)$, and the folding integer vector of $a$ is the
desired output.

It remains to note that, for general moduli, the system
\eqref{eq:chosen_remainder_crt_system} may not have a solution. 
This is because the full remainder space has size $\prod_{i=1}^L m_i$, while the number of remainder vectors coming from actual integers modulo the lcm $M=\lcm(m_1,\ldots,m_L)$ is only $M$.
Hence, when moduli $m_i$, $1\leq i \leq L$, are not pairwise co-prime, then $M< \prod_{i=1}^L m_i$. 
Thus, not every remainder tuple in the remainder space is a valid vector of remainders. By the general Chinese remainder theorem \cite{Ore52}, the system
\eqref{eq:chosen_remainder_crt_system} is solvable if and only if
\begin{equation}\label{eq:general_crt_compatibility}
\rho_i\equiv \rho_j \mod{\gcd(m_i,m_j)},
\qquad
1\le i,j\le L.
\end{equation}
When this condition holds, the solution is unique modulo
$M=\lcm(m_1,\ldots,m_L)$.

In an implementation, one does not need to check all pairwise conditions in
\eqref{eq:general_crt_compatibility} separately before applying the CRT.
Instead, the congruences can be merged successively. Suppose that after some
indices have been merged, we have a congruence
\[
x\equiv c \mod{M_c}.
\]
To merge it with the next congruence $x\equiv \rho_i \mod {m_i}$, we only need to check
\[
c\equiv \rho_i \mod{\gcd(M_c,m_i)}.
\]
If this condition fails, then the current tuple is invalid and cannot correspond
to any integer. If it holds, the two congruences are merged into one congruence
modulo $\lcm(M_c,m_i)$. Repeating this process either discards the tuple at some
merge step or produces an integer $a\in[0,M-1]_{\mathbb Z}$. We then check
whether this integer lies in $[0,D-1]_{\mathbb Z}$.

To reduce the searching complexity, we choose an index subset
$\mathcal S\subseteq\{1,\ldots,L\}$ whose moduli satisfy
\[
M_{\mathcal S}\triangleq \lcm\{m_i:i\in\mathcal S\}\ge D,
\]
and for which the product $\prod_{i\in\mathcal S}|\mathcal R_i|$ is small.
The condition $M_{\mathcal S}\ge D$ ensures that the remainders on the indices
in $\mathcal S$ determine at most one integer in $[0,D-1]_{\mathbb Z}$.
Therefore, after a remainder tuple on $\mathcal S$ is fixed, solving the subsystem
of \eqref{eq:chosen_remainder_crt_system} restricted to $i\in\mathcal S$ gives
either no candidate in $[0,D-1]_{\mathbb Z}$ or a unique candidate integer.
More precisely, if the subsystem has no solution, then the tuple is discarded.
If it has a solution, let $a\in[0,M_{\mathcal S}-1]_{\mathbb Z}$ be the
integer obtained from the CRT. If $a\ge D$, then it is outside the
candidate range and is discarded. If $a<D$, then $a$ is a candidate obtained from the indices in
$\mathcal S$. We then verify it on the remaining indices by checking whether
$a\bmod m_i\in\mathcal R_i$,
for all $i\in\{1,\ldots,L\}\setminus\mathcal S$.
If this check holds, then $a$ is feasible for
$\widetilde{\mathbf r}$ under the remainder error bound vector $\boldsymbol{\tau}$, and the folding integer vector of $a$ is
output. Otherwise, $a$ is discarded and the algorithm continues to the next
tuple. Choosing $\mathcal S$ with a small product
$\prod_{i\in\mathcal S}|\mathcal R_i|$ reduces the number of remainder tuples that
need to be tested.

The complete procedure is summarized in Algorithm~\ref{alg:exact_decoder}. 
The correctness follows from two observations. First, suppose the observed
remainder vector $\widetilde{\mathbf r}$ is produced from some
$N\in[0,D-1]_{\mathbb Z}$ with remainder errors bounded by
$\boldsymbol{\tau}$. Then the true remainders of $N$ on the selected indices in
$\mathcal S$ belong to the corresponding sets $\mathcal R_i$, and hence appear
among the enumerated remainder tuples. Therefore, the algorithm can
find at least one feasible integer in $[0,D-1]_{\mathbb Z}$ for the observation
$\widetilde{\mathbf r}$. Second, when $\Delta(\boldsymbol{\tau},D)>0$, any two
feasible integers in $[0,D-1]_{\mathbb Z}$ for the same observation must have
the same folding integer vector, since otherwise the same observation would
belong to two observation sets corresponding to distinct folding integer
vectors. Therefore, the folding integer vector of the first feasible candidate
found by the algorithm is correct. This gives the following result.

\begin{algorithm}[t]
\caption{Folding integer vector decoding algorithm}
\label{alg:exact_decoder}
\begin{algorithmic}[1]
\REQUIRE Moduli $m_1,\ldots,m_L$, range length $D$, remainder error bound vector
$\boldsymbol{\tau}$, and erroneous remainder vector $\widetilde{\mathbf r}$.
\ENSURE A folding integer vector, or the empty set.

\STATE Form the sets $\mathcal R_i$ in
\eqref{eq:possible_true_remainder_set}, $i=1,\ldots,L$.

\STATE Choose a subset $\mathcal S\subseteq\{1,\ldots,L\}$ such that
$M_{\mathcal S}\triangleq \lcm\{m_i:i\in\mathcal S\}\ge D$, with small product
$\prod_{i\in\mathcal S}|\mathcal R_i|$.

\FOR{each tuple $(\rho_i)_{i\in\mathcal S}\in\prod_{i\in\mathcal S}\mathcal R_i$}

    \IF{the subsystem $x\equiv \rho_i \mod {m_i}$, $i\in\mathcal S$, has no integer solution}
        \STATE Continue to the next tuple.
    \ENDIF

    \STATE Solve the subsystem by the CRT and let $a\in[0,M_{\mathcal S}-1]_{\mathbb Z}$ be the solution.

    \IF{$a\ge D$}
        \STATE Continue to the next tuple.
    \ENDIF

    \IF{$a\bmod m_i\in\mathcal R_i$ for all $i\in\{1,\ldots,L\}\setminus\mathcal S$}
        \STATE Return
        $\left(
        \left\lfloor\frac{a}{m_1}\right\rfloor,\ldots,
        \left\lfloor\frac{a}{m_L}\right\rfloor
        \right)$.
    \ENDIF

\ENDFOR

\STATE Return $\emptyset$.
\end{algorithmic}
\end{algorithm}

\begin{thm}\label{thm:exact_decoder_correctness}
Assume that $\Delta(\boldsymbol{\tau},D)>0$. Suppose that the observed remainder
vector $\widetilde{\mathbf r}$ is produced from an integer
$N\in[0,D-1]_{\mathbb Z}$ in the sense that there exists
$\Delta\mathbf r=(\Delta r_1,\ldots,\Delta r_L)\in\mathcal E_{\boldsymbol{\tau}}$
such that
$\widetilde r_i=N\bmod m_i+\Delta r_i$, for
$i=1,\ldots,L$.
Then Algorithm~\ref{alg:exact_decoder} outputs the true folding integer vector
$\mathbf n(N)$.
\end{thm}

We now discuss the complexity. Let $R_i=|\mathcal R_i|$. From
\eqref{eq:possible_true_remainder_set}, $R_i\le \min\{m_i,2\tau_i+1\}$ for
$i=1,\ldots,L$. For the selected subset $\mathcal S$, the number of remainder
tuples to be tested is at most $\prod_{i\in\mathcal S}R_i$. The algorithm may
stop earlier, since it stops once a verified candidate is found. Let $T$ be
the number of tested tuples before termination. Then
$T\le\prod_{i\in\mathcal S}R_i$. For each tested tuple, the CRT subsystem on
$\mathcal S$ is solved by successive CRT merges, and a candidate integer, if
obtained, is verified on the remaining indices. Hence, ignoring the bit
complexity of integer arithmetic, each tested tuple costs at most $O(L)$
operations, and the running time is $O(TL)$. In the worst case, the complexity
is bounded by
\begin{equation}\label{eq:decoder_worst_complexity_bound}
O\left(L\prod_{i\in\mathcal S}R_i\right).
\end{equation}
Compared with a direct search over all folding integer vectors in $\mathcal F(D)$, the proposed algorithm searches only over the possible remainders on the selected indices allowed by the observation $\widetilde{\mathbf r}$. This can be much lower when the remainder error bounds are small, because $R_i\le 2\tau_i+1$ and $\prod_{i\in\mathcal S}R_i$ can be much smaller than $|\mathcal F(D)|$, especially when $D$ is large.

\section{Special Cases and Connections}
\label{s5}

We now discuss several special cases. These cases help
interpret the general characterization and check its consistency with the
known CRT and robust CRT settings. We first consider the error-free case
and then specialize the
theory to uniform remainder error bounds, where the error bound vector has the
form $\boldsymbol{\tau}=\tau\mathbf 1$ and the general condition can be written in terms of the common error bound $\tau$.

\subsection{The Error-Free Case}

We first consider the case with no remainder errors, namely
$\boldsymbol{\tau}=\mathbf 0$. In this case, robust determination is the same as ordinary determination.
For consistency, we keep the term ``robust determination'' in the following studies.

\begin{cor}\label{cor:zero_error_full_range}
When $\boldsymbol{\tau}=\mathbf 0$, we have
$D(\mathbf 0)=M$.
\end{cor}

\begin{proof}
When $\boldsymbol{\tau}=\mathbf 0$, the admissible error set is
$\mathcal E_{\mathbf 0}=\{\mathbf 0\}$. Hence the observed remainder vector is
exactly the true remainder vector. The proof then follows from the 
conventional CRT.
\end{proof}

Corollary~\ref{cor:zero_error_full_range} shows that the proposed tradeoff reduces to the classical CRT dynamic range when there is no remainder error.

\subsection{The Uniform Remainder Error Bound Case}
\label{s5_uniform}

We next consider the uniform remainder error bound case, where
$\boldsymbol{\tau}=\tau\mathbf 1$,
for some $\tau\in\mathbb Z_{\ge0}$.
The purpose of this subsection is to show that the proposed general 
framework reduces to the known uniform dynamic range--robustness tradeoff in
\cite[Th. 3 and Cor. 3]{XiaoHuangYeXiao18}, with the boundary modified by
the integer closed-bound error model used in this paper.

Let
\begin{equation}\label{eq:m_min_uniform_def}
m_{\min}=\min_{1\le i\le L}m_i.
\end{equation}
For $x\in[0,M-1]_{\mathbb Z}$, define
\begin{equation}\label{eq:psi_uniform_def}
\psi(x)
=
\max_{1\le i\le L}(x\bmod m_i).
\end{equation}
Thus, $\psi(x)$ is the largest remainder of $x$ among the moduli
$m_1,\ldots,m_L$. Then, it is not hard to derive the following corollary from Theorem~\ref{thm:exact_fixed_D_tradeoff}.

\begin{cor}
\label{cor:uniform_error_reduction}
Let $\boldsymbol{\tau}=\tau\mathbf 1$. For every $1\le D\le M$, the robust
determination on $[0,D-1]_{\mathbb Z}$ is possible under the integer error
bound $|\Delta r_i|\le\tau$ if and only if
\begin{equation}\label{eq:uniform_fixed_D_condition}
\psi(x)>4\tau
\quad
\text{for every }x\in[m_{\min},D-1]_{\mathbb Z}.
\end{equation}
\end{cor}

Corollary~\ref{cor:uniform_error_reduction} is the integer error counterpart
of the uniform tradeoff in \cite[Th. 3 and Cor. 3]{XiaoHuangYeXiao18}. In
\cite{XiaoHuangYeXiao18}, the dynamic range is written as an endpoint $K$ of
the interval $[0,K]$, while this paper uses the range length $D$ of
$[0,D-1]_{\mathbb Z}$. Thus the correspondence is $K=D-1$.

The main difference is the boundary condition. The result in
\cite{XiaoHuangYeXiao18} is stated for a real-valued remainder error bound $\delta$ with
strict error bounds $|\Delta r_i|<\delta$, which leads to the condition
$\psi(x)\ge 4\delta$
over the considered range. In this paper, the remainder errors are integers and
the error model uses the closed bounds $|\Delta r_i|\le\tau$. Therefore, when
$\psi(x)=4\tau$, two different folding integer vectors can generate
the same erroneous remainder vector with admissible integer errors. Hence the
correct condition under the present model is the strict inequality
$\psi(x)>4\tau$.

For the two-modulus case, i.e., $L=2$, Corollary~\ref{cor:uniform_error_reduction} further
reduces to the integer error counterpart of the two-modulus dynamic
range--robustness tradeoff obtained in \cite{XiaoXiaHuo17}.

\section{Numerical Results}
\label{s6}

In this section, we illustrate the derived exact tradeoff from three
perspectives. First, we verify the guaranteed recovery property of the proposed decoding
algorithm for folding integer vectors. Within the dynamic range, the folding integer vector is correctly recovered
for every integer and all admissible remainder errors.
Second, we show that pairwise co-prime moduli can
still provide robustness when the dynamic range is reduced, even when every
remainder may contain a bounded error. 
Third, in addition to the two examples in Section~\ref{s3}, we further
demonstrate why the full remainder error bound vector is necessary. Replacing a
non-uniform error bound vector by a single uniform bound can severely
underestimate the achievable dynamic range.

\subsection{Guaranteed Recovery and Searching Complexity of the Decoding Algorithm}

We first test the decoding algorithm in Section~\ref{s4}. We use
$(m_1,m_2,m_3)=(65,120,162)$ and
$\boldsymbol{\tau}=(1,2,1)$.
Theorem~\ref{thm:exact_D_tau} gives $D(\boldsymbol{\tau})=10206$.
For each tested range length $D$, we randomly generate integers
$N\in[0,D-1]_{\mathbb Z}$ and integer remainder errors satisfying
$|\Delta r_i|\le\tau_i$, $i=1,2,3$. We then apply Algorithm~\ref{alg:exact_decoder} and record the probability of recovering the correct folding integer vector. We also record the number $T$ of remainder tuples tested before the algorithm terminates.

\begin{figure}[htbp]
\centering
\includegraphics[width=0.3\textwidth]{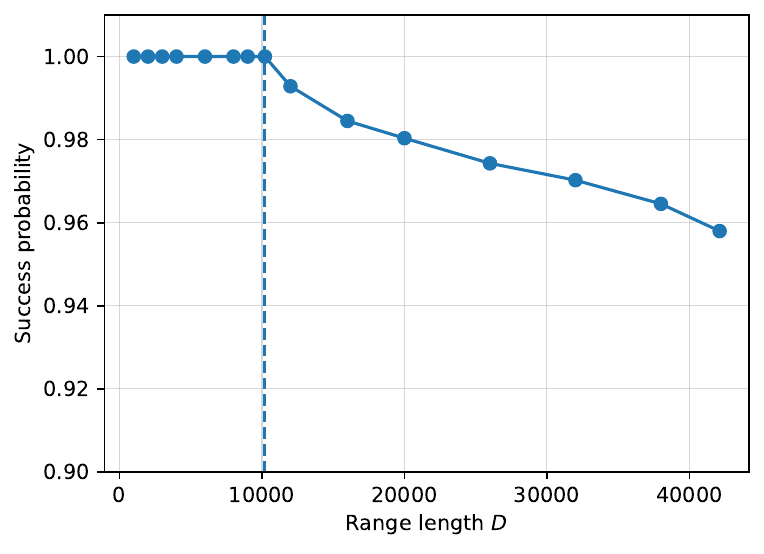}
\caption{Probability of correct folding integer vector recovery by Algorithm~\ref{alg:exact_decoder} for $(m_1,m_2,m_3)=(65,120,162)$ and $\boldsymbol{\tau}=(1,2,1)$. The vertical dotted line marks the dynamic range $D(\boldsymbol{\tau})=10206$.}
\label{fig:decoder_success_vs_D}
\end{figure}

Fig.~\ref{fig:decoder_success_vs_D} verifies the guaranteed recovery within the dynamic range. When $D\le D(\boldsymbol{\tau})=10206$, Theorem~\ref{thm:exact_fixed_D_tradeoff} guarantees correct folding integer vector recovery for every integer in $[0,D-1]_{\mathbb Z}$ and every admissible remainder error vector. The simulation agrees with this guarantee. When $D>D(\boldsymbol{\tau})$, the guarantee no longer holds, and the empirical success probability can drop below one.

\begin{figure}[htbp]
\centering
\includegraphics[width=0.3\textwidth]{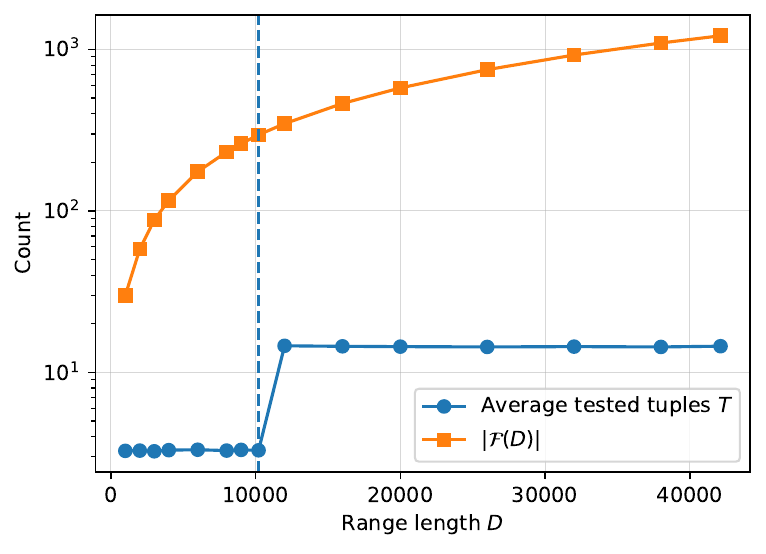}
\caption{Average number of tested remainder tuples $T$ in Algorithm~\ref{alg:exact_decoder}, compared with the number of folding integer vectors $|\mathcal F(D)|$. The vertical dotted line marks $D(\boldsymbol{\tau})=10206$.}
\label{fig:decoder_cost_vs_D}
\end{figure}

Fig.~\ref{fig:decoder_cost_vs_D} shows the searching complexity of the decoding algorithm. Compared with a direct search over all folding integer vectors in $\mathcal F(D)$, the proposed algorithm searches only over the possible true remainders on the selected indices allowed by the observation $\widetilde{\mathbf r}$. At $D=D(\boldsymbol{\tau})$, we have $|\mathcal F(D)|=296$, while the average number of tested remainder tuples is much smaller in the simulation. This is because the algorithm can choose $\mathcal S=\{1,3\}$, for which $\lcm(65,162)=10530\ge D(\boldsymbol{\tau})$ and $R_1R_3\le 9$.

\subsection{Robustness with Pairwise Co-Prime Moduli}

We next give an example showing that robustness can still be obtained
with pairwise co-prime moduli after reducing the dynamic range. Let
$(m_1,m_2,m_3,m_4)=(7,11,17,19)$ and $\boldsymbol{\tau}=(2,1,1,1)$.
The moduli are pairwise co-prime and
$M=7\cdot 11\cdot 17\cdot 19=24871$. By
Theorem~\ref{thm:exact_D_tau}, we obtain
$D(\boldsymbol{\tau})=136$.
Thus, every integer in $[0,135]_{\mathbb Z}$ has its folding integer vector
robustly determined, even though all four remainders may be erroneous.

\subsection{Dynamic Range Loss under Uniform Remainder Error Bound Replacement}

Finally, we show why it is important to keep the full vector
$\boldsymbol{\tau}$ instead of replacing it by a scalar uniform bound. We
again use
$(m_1,m_2,m_3)=(6,9,15)$
and test all remainder error bound vectors in $[0,5]_{\mathbb Z}^3$.

\begin{figure}[htbp]
\centering
\includegraphics[width=0.3\textwidth]{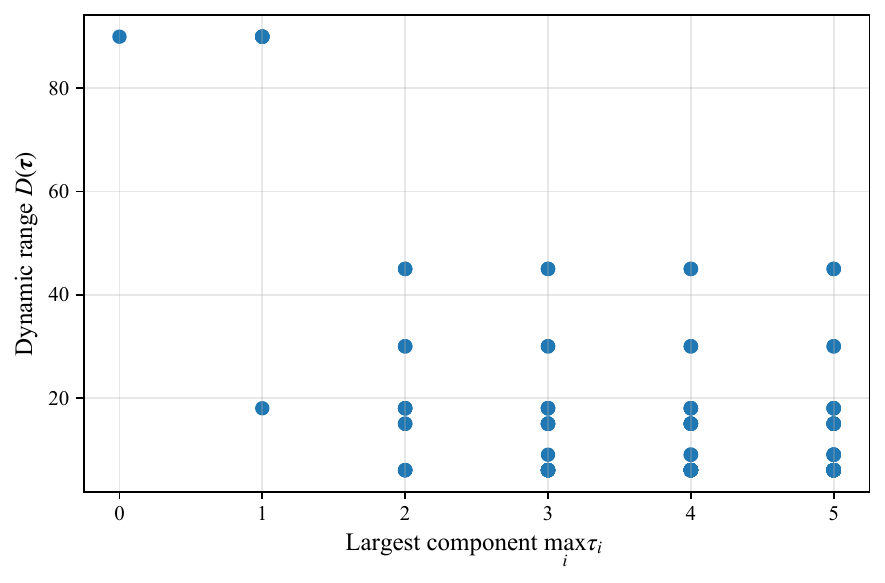}
\caption{Dynamic range $D(\boldsymbol{\tau})$ versus the largest error
bound component $\max_i\tau_i$ for $(m_1,m_2,m_3)=(6,9,15)$ and
$\boldsymbol{\tau}\in[0,5]_{\mathbb Z}^3$. Each dot represents one remainder error bound vector
$\boldsymbol{\tau}\in[0,5]_{\mathbb Z}^3$.}
\label{fig:nonuniform_dynamic_range_scatter}
\end{figure}

Fig.~\ref{fig:nonuniform_dynamic_range_scatter} plots the exact dynamic
range $D(\boldsymbol{\tau})$ against that with the uniform remainder error bound of the largest component $\max_i\tau_i$.
Remainder error bound vectors with the same value of $\max_i\tau_i$ can have very
different dynamic ranges. Therefore, the dynamic range depends not only
on the largest remainder error bound, but also on how the remainder error bounds are
distributed across the remainder channels.

To quantify the loss caused by replacing $\boldsymbol{\tau}$ with the uniform
vector $(\max_i\tau_i)\mathbf 1$, define
\[
G(\boldsymbol{\tau})
=
D(\boldsymbol{\tau})
-
D\bigl((\max_i\tau_i)\mathbf 1\bigr).
\]
The quantity $G(\boldsymbol{\tau})$ measures the loss in dynamic range caused
by replacing the actual non-uniform remainder error bound vector with the uniform bound equal to its largest component. This loss reflects the price of ignoring the remainder
channel-dependent error bounds. 

\begin{figure}[htbp]
\centering
\includegraphics[width=0.3\textwidth]{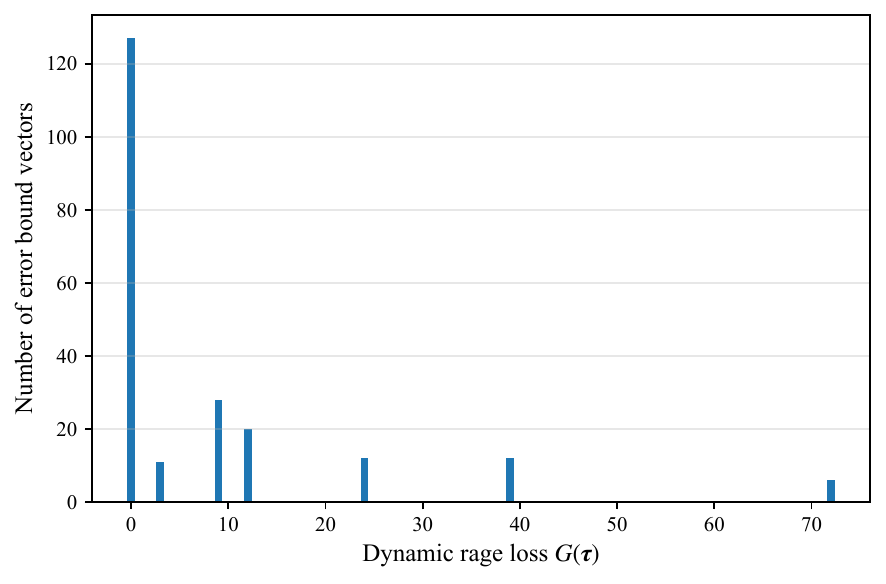}
\caption{Distribution of the dynamic range loss
$G(\boldsymbol{\tau})$ over all remainder error bound vectors
$\boldsymbol{\tau}\in[0,5]_{\mathbb Z}^3$.}
\label{fig:nonuniform_gain_hist}
\end{figure}

Fig.~\ref{fig:nonuniform_gain_hist} shows the distribution of $G(\boldsymbol{\tau})$ over all tested remainder error bound vectors. 
The loss is not uniform across the remainder error bound space. 
Many vectors have $G(\boldsymbol{\tau})=0$, which means that replacing them by $(\max_i\tau_i)\mathbf 1$ does not reduce the dynamic range.
However, for some vectors, the loss can be large. In the tested example, the largest observed loss is $72$.
For instance, every non-uniform remainder error bound vector in $\{0,1\}^3$ has full dynamic range $90$, while with the uniform replacement as $(1,1,1)$, the dynamic range is significantly reduced to $18$. 

These results show that the robustness under non-uniform remainder errors is coupled across remainder channels. 
For a fixed dynamic range, the tolerable remainder error bound in one channel may depend on the error bounds in the other channels. 
Therefore, the full remainder error bound vector is needed to describe the exact dynamic range--robustness tradeoff.

\section{Conclusion}
\label{s7}

This paper studied the exact dynamic range--robustness tradeoff for CRT when the remainder error bound is a vector rather than a single scalar.
We viewed robust CRT as an exact folding integer vector decoding problem and converted the robust determination into a disjointness problem between observation sets of erroneous remainders.
Based on this viewpoint, we derived a scalar necessary and sufficient condition for robust determination at any fixed range length and remainder error bound vector.
We then characterized the dynamic range for a given remainder error bound vector and the tolerable remainder error bound vectors for a given range length.
We also gave an exact folding integer vector decoding algorithm.
The special cases of error-free and uniform error bounds coincide with the known CRT and robust CRT results, respectively.
Numerical results showed that the dynamic range depends on the full remainder error bound vector. 
They also showed that remainder channels are coupled in the robustness tradeoff, so a single uniform remainder error bound cannot fully describe the effect of non-uniform remainder errors.
Therefore, the proposed characterization can be used as a design tool for deciding how much error each remainder channel may tolerate and how the moduli should be selected when different channels have different noise levels or reliability requirements.

\end{document}